\documentclass[conference,compsoc]{IEEEtran}

\ifCLASSOPTIONcompsoc
  \usepackage[nocompress]{cite}
\else
  \usepackage{cite}
\fi

\ifCLASSINFOpdf
\else
\fi

\usepackage[color,secparam=\kappa]{crypto}

\usepackage{tikz}
\usepackage{framed}
\usepackage{tabularx}
\usepackage{multirow}
\usepackage{amsthm}
\usepackage{amsmath}
\usepackage{enumitem}
\usepackage{amsfonts}
\usepackage{graphicx}
\usepackage{stmaryrd}
\usepackage[hidelinks]{hyperref}
\usepackage{algorithm2e}
\usepackage{subcaption}
\usepackage{diagbox}
\usepackage[table]{xcolor}
\usepackage{xcolor}
\usepackage{tcolorbox}
\usetikzlibrary{shapes.arrows}

\usetikzlibrary{
    positioning,
    arrows.meta,
    calc,
    fit
}

\usepackage[available,functional,reproduced]{ieeebadges}

\newtheorem{definition}{Definition}
\newtheorem{theorem}{Theorem}

\newcommand{\receiver}{$\mathcal{R}$ }
\newcommand{\sender}{$\mathcal{S}$ }

\begin{document}
\title{Efficient Fuzzy Private Set Intersection from Secret-shared OPRF}

\author{
\IEEEauthorblockN{
Xinpeng Yang\textsuperscript{1},
Meng Hao\textsuperscript{2}\textsuperscript{*},
Chenkai Weng\textsuperscript{3}, 
Robert H. Deng\textsuperscript{2},
Yonggang Wen\textsuperscript{1},
Tianwei Zhang\textsuperscript{1}
}
\IEEEauthorblockA{\textsuperscript{1}Nanyang Technological University, xinpeng004@e.ntu.edu.sg, \{ygwen, tianwei.zhang\}@ntu.edu.sg \\ \textsuperscript{2}Singapore Management University, menghao303@gmail.com, robertdeng@smu.edu.sg \\
\textsuperscript{3}Arizona State University, chenkai.weng@asu.edu}
}

\maketitle

\makeatletter
\renewcommand\@makefntext[1]{\noindent\makebox[0em][r]{\@makefnmark\ }#1}
\makeatother

\begingroup
\renewcommand\thefootnote{}
\footnotetext{\IEEEauthorrefmark{1} Meng Hao is the corresponding author}
\endgroup

\begin{abstract}
Private set intersection (PSI) enables a sender holding a set $Q$ of size $m$ and a receiver holding a set $W$ of size $n$ to securely compute the intersection $Q \intersection W$. Fuzzy PSI (FPSI) is a PSI variant where the receiver learns the items $q \in Q$ for which there exists some $w \in W$ satisfying $\mathsf{dist}(q, w) \le \delta$ under a given distance metric.
Although several FPSI works are proposed for $L_{p}$ distance metrics with $p \in [1, \infty]$, they either heavily rely on expensive homomorphic encryptions, or incur undesirable complexity, e.g., exponential to the element dimension, both of which lead to poor practical efficiency.

In this work, we propose efficient FPSI protocols for $L_{p \in [1, \infty]}$ distance metrics, primarily leveraging significantly cheaper symmetric-key operations.
Our protocols achieve linear communication and computation complexity in the set sizes $m,n$, the dimension $d$, and the distance threshold $\delta$.
Our core building block is an oblivious programmable PRF with secret-shared outputs, which may be of independent interest. 
Furthermore, we incorporate a prefix technique that reduces the dependence on the distance threshold $\delta$ to logarithmic, which is particularly suitable for large $\delta$.

We implement our FPSI protocols and compare them with state-of-the-art constructions. Experimental results demonstrate that our protocols consistently and significantly outperform existing works across all settings. Specifically, our protocols achieve a speedup of $12{\sim}145\times$ in running time and a reduction of $3{\sim}8\times$ in communication cost compared to Gao et al.~(ASIACRYPT'24) and a speedup of $9{\sim}80\times$ in running time and a reduction of $5{\sim}19\times$ in communication cost compared to Dang et al.~(CCS'25).

\end{abstract}

\IEEEpeerreviewmaketitle

\section{Introduction}
\label{sec: Introduction}

Private set intersection (PSI)~\cite{meadows1986more,freedman2004efficient,dong2013private,pinkas2015phasing,pinkas2020psi,raghuraman2022blazing,hao2024unbalanced} enables two parties, the sender and the receiver, to compute the intersection of their sets without revealing additional information beyond the intersection itself.
Standard PSI protocols focus on exact matching and have been widely employed in various scenarios such as password breach monitoring and genome matching~\cite{google-psi-password,shen2018efficient}.
However, in some complicated applications with multi-dimensional inputs, requiring exact intersection is often impractical or even impossible.
For example, in biometric authentication systems (e.g., fingerprint or face), two samples belonging to the same individual may differ due to environmental variations and feature extraction perturbations~\cite{uzun2021fuzzy}, which limits the application of standard PSI protocols.

To address this problem, fuzzy PSI extends PSI by allowing the receiver, holding the set $W$, to learn those items $q \in Q$ from the sender for which there exists some $w \in W$ satisfying $\mathsf{dist}(q, w) \le \delta$ under a specified distance metric.  
Recently, a substantial body of works~\cite{garimella2022structure,garimella2024computation,van2024fuzzy,gao2025efficient,zhang2025fast,dang2025ccs,piske2025distance,van2025,richardson2024fuzzy,bui2025new} have proposed a variety of fuzzy PSI constructions.  
Among these, van Baarsen and Pu~\cite{van2024fuzzy} introduced the first generic fuzzy PSI protocols supporting arbitrary $L_{p\in[1,\infty]}$ distance. Although several follow-up works~\cite{van2025,piske2025distance} further improve performance, these protocols still incur super-linear complexity, e.g., exponential in the element dimension, which becomes prohibitively expensive in high-dimensional settings.

More recently, Gao et al.~\cite{gao2025efficient} proposed the first fuzzy PSI protocols for $L_{p\in[1,\infty]}$ distance with linear complexity in the set size, the dimension, and the distance threshold. Dang et al.~\cite{dang2025ccs} further optimized this line of work by introducing the prefix representation~\cite{chakraborti2023distance} that reduces the dependence on the threshold to logarithmic, thereby reducing both computation and communication costs, particularly for large thresholds.  
Nevertheless, both schemes rely heavily on expensive additively homomorphic encryption operations, such as the ElGamal and Paillier schemes, which leads to poor concrete performance. Please refer to Table~\ref{tab:complexities} and Section~\ref{sec: Related Work} for a more detailed discussion of related works.

\begin{table*}[t]
\centering
\caption{Asymptotic complexities of existing fuzzy PSI protocols for $L_{p \in [1, \infty]}$ distance, where the sender holds $m$ elements and the receiver holds $n$ elements in a $d$-dimensional space, $\delta$ is the distance threshold. We ignore multiplicative factors of the computational security parameter $\kappa$ and statistical security parameter $\lambda$.}
\label{tab:complexities}
\resizebox{\textwidth}{!}{%
\begin{tabular}{|c|c|c|c|cc|}
\hline
\multirow{2}{*}{Metric}                  & \multirow{2}{*}{Protocol}                      & \multirow{2}{*}{Assumption}                         & \multirow{2}{*}{Communication}                         & \multicolumn{2}{c|}{Computation}                                                                                          \\ \cline{5-6} 
                                         &                                                &                                                     &                                                        & \multicolumn{1}{c|}{Sender}                                                 & Receiver                                    \\ \hline \hline
\multirow{10}{*}{$L_\infty$}              
                                         & \multirow{3}{*}{\cite{van2024fuzzy}}           & $\mathcal{R}, \min > 2\delta$                       & $O(\delta d n + 2^d m)$                                & \multicolumn{1}{c|}{$O(2^dd m)$}                                & $O(\delta d n + 2^d m)$                     \\
                                         &                                                & $\mathcal{R}, \min > 4\delta$                       & $O(\delta2^ dd n + m)$                                  & \multicolumn{1}{c|}{$O(dm)$}                                                & $O(\delta2^dd n + m)$                       \\
                                         &                                                & $\mathcal{R}, \text{disj. proj.}$                    & $O((\delta d)^2 n + m)$                                & \multicolumn{1}{c|}{$O(d^2 m)$}                                & $O((\delta d)^2 n + m)$                                  \\ \cline{2-6} 
                                         & \cite{bui2025new}                & $\mathcal{R},\text{mini-univ.}$ & $O((dn \log \delta + m(2\log \delta)^d))$              & \multicolumn{1}{c|}{$O(m(2\log \delta)^d)$}             & $O((\log \delta)^d n + dm\log \delta)$ \\ \cline{2-6}
                                         & \cite{van2025}                        & $\mathcal{R} \wedge \mathcal{S},\text{ min}>4\delta$ & $O(d(m +2^dn)\log\delta)$                                    & \multicolumn{1}{c|}{$O(d(m +2^dn)\log\delta)$}                                    & $O(2^dn(d\log\delta+(\log\delta)^{d/2}))$                         \\ \cline{2-6}
                                        & \cite{zhang2025fast}                        & $\mathcal{R} \wedge \mathcal{S},s\text{-separate}$ & $O(\delta^s\frac{d}{s}(m+n))$                                    & \multicolumn{1}{c|}{$O(\delta^s\frac{d}{s}m+n)$}                                    & $O(\delta^s\frac{d}{s}n+m)$                         \\ \cline{2-6}
                                         & \cite{piske2025distance}                        & $\mathcal{S}, \text{disj. hash}$ & $O(d (\delta m+2^dn))$                                    & \multicolumn{1}{c|}{$O(\delta d m)$}                                    & $O(d 2^dn + m)$                         \\ \cline{2-6}
                                         & \cite{richardson2024fuzzy}                     & $\mathcal{S}, \text{disj. hash}$                 & $O(d \log \delta (n 2^d + m 2^{d-s}))$            & \multicolumn{1}{c|}{$O(2^{d-s} dm \log \delta)$}                               & $O(2^{s} dn \log \delta)$           \\ \cline{2-6}
                                         & \cite{gao2025efficient}                        & $\mathcal{R} \wedge \mathcal{S}, \text{disj. proj.}$ & $O(\delta d (m+n))$                                    & \multicolumn{1}{c|}{$O(\delta d m + n)$}                                    & $O(\delta d n + m)$                         \\ \cline{2-6} 
                                         & \cite{dang2025ccs}                                           &   $\mathcal{R} \wedge \mathcal{S}, \text{disj. proj.}$                                                  &   $\bigo{ d(m+n)\log\delta}$                                                      & \multicolumn{1}{c|}{$\bigo{ d(m+n)\log\delta}$ }                                                       & $\bigo{ d(m+n)\log\delta}$
                                         \\ \cline{2-6}
                                        & Ours                                           &   $\mathcal{R} \wedge \mathcal{S}, \text{disj. proj.}$                                                  &   $\bigo{\delta d(m+n)}$                                                     & \multicolumn{1}{c|}{$\bigo{\delta dm+dn}$}                                                       &   $\bigo{\delta dn+dm}$
                                        
                                         \\ \cline{2-6}
                                         
                                         & Ours-prefix                                    &     $\mathcal{R} \wedge \mathcal{S}, \text{disj. proj.}$                                                &     $\bigo{ d(m+n)\log\delta}$                                                   & \multicolumn{1}{c|}{$\bigo{ d(m+n)\log\delta}$}                                                       &    $\bigo{ d(m+n)\log\delta}$
                                                                                   
                                                                                  \\ \hline \hline
\multirow{6}{*}{$L_{p \in [1, \infty)}$} & \multirow{2}{*}{ \cite{van2024fuzzy}}           & $\mathcal{R}, \min > 2\delta(d^{1/p}+1)$             & $O(\delta{2^d} d n + \delta^p m)$                   & \multicolumn{1}{c|}{$O((d + \delta^p) m)$}                   & $O(\delta{2^d} d n + m)$                      \\
                                         &                                                & $\mathcal{R}, \min > \delta/\rho$                                                 & $O(\delta d n^{1+\rho} + \delta^\rho m n^\rho \log n)$ & \multicolumn{1}{c|}{$O((d + \delta^\rho) m n^\rho \log n)$ } & $O(\delta d n^{1+\rho} +  m n^\rho \log n)$        \\ \cline{2-6}
                                         & \cite{van2025}                        & $\mathcal{R} \wedge \mathcal{S},\text{ min}>2\delta(d^{1/p}+1)$ & $O(d (\delta m+2^dn)+p(m+2^dn)\log\delta)$                                    & \multicolumn{1}{c|}{$O(d (\delta m+2^dn)+p(m+2^dn)\log\delta)$}                                    & $O(2^dn(d+p\log\delta))$                         \\ \cline{2-6}
                                         & \cite{zhang2025fast}                        & $\mathcal{R} \wedge \mathcal{S},s\text{-separate}$ & $O(\delta^s\frac{d}{s}(m+n)+pm\log\delta)$                                    & \multicolumn{1}{c|}{$O((\delta^s\frac{d}{s}+p\log\delta)m+n)$}                                    & $O(\delta^s\frac{d}{s}n+pm\log\delta)$                         \\ \cline{2-6}
                                         & \cite{piske2025distance}                        & $\mathcal{S},\text{ disj. hash}$ & $O(d (\delta m+2^dn))$                                    & \multicolumn{1}{c|}{$O(\delta d m)$}                                    & $O(d 2^dn + m)$                         \\ \cline{2-6}
                                         & \cite{gao2025efficient}                        & $\mathcal{R} \wedge \mathcal{S}, \text{disj. proj.}$ & $O(\delta d (m+n)+pm\log\delta)$                                    & \multicolumn{1}{c|}{$O(\delta d m+pm\log\delta + n)$}                                    & $O(\delta d n+pm\log\delta)$                         \\ \cline{2-6}
                                         & \cite{dang2025ccs}                                           &   $\mathcal{R} \wedge \mathcal{S}, \text{disj. proj.}$                                                  &   $\bigo{dpn\log\delta+dm\log\delta}$                                                     & \multicolumn{1}{c|}{$\bigo{dpm\log\delta+dn\log\delta}$}                                                       & $\bigo{dpn\log\delta+dm\log\delta}$                                           \\ \cline{2-6}

                                         & Ours                                           & $\mathcal{R} \wedge \mathcal{S}, \text{disj. proj.}$                                                    &   $\bigo{\delta d(m+n)+pm\log\delta}$                                                     & \multicolumn{1}{c|}{$\bigo{\delta dm+dn+pm\log\delta}$}                                                       &  $\bigo{\delta dn+dm+pm\log\delta}$
                                         \\ \cline{2-6}
                                         & Ours-prefix                                    &   $\mathcal{R} \wedge \mathcal{S}, \text{disj. proj.}$                                                  &   $\bigo{dpn\log\delta+dpm\log\delta}$                                                     & \multicolumn{1}{c|}{$\bigo{dpm\log\delta+dpn\log\delta}$}                                                       & $\bigo{dpn\log\delta+dpm\log\delta}$
                                                                                    
                                         \\ \hline
\end{tabular}%
}
\captionsetup{justification=raggedright, singlelinecheck=false}
\caption*{\footnotesize 
-- $\mathcal{R}/\mathcal{S}$ denotes that the set of receiver/sender satisfies the assumption, and $\mathcal{R} \wedge \mathcal{S}$ indicates that both sets satisfy the assumption.

-- disj. proj. is short for disjoint projection assumption, where each element has disjoint projections in some dimension with all other elements.

-- disj. hash is short for disjoint hash assumption, which means the spatial hashing scheme maps at most one point to every possible target grid cell.

-- mini-univ is short for mini-universe assumption, where every ball can be mapped to a distinct vertex of the space.

--  $s$-separate means for every $s$ dimensions, elements have disjoint projections in one of these $s$ dimensions with all other elements where $1\le s \le d$.

-- $0 < \rho < 1/c$ is a parameter in locality-sensitive hashing where the distance between any receiver's two elements is greater than $c\delta$.

-- $\min > l$ means that the minimum distance between any two elements of the set is greater than $l$.

}
\end{table*}

Motivated by the above, we raise the following question:

\vspace{1pt}
\begin{center}
        \begin{tcolorbox}[
            colback=lightgray!15, 
            size=title, 
            colframe=black, 
            boxrule=0.5pt,
            width=\linewidth,
            valign=center,
            ]
        
            Can we construct concretely efficient fuzzy PSI protocols for general $L_{p\in[1, \infty]}$ distance with linear complexity in the set size and the element dimension?
          
        \end{tcolorbox}
\end{center}
\vspace{0.5pt}

\subsection{Our Contribution}
We answer the above question affirmatively and summarize our contributions as follows.
\begin{enumerate} 

    \item \textbf{Shared-output OPPRF.} We introduce a new building block, termed oblivious programmable PRF with secret-shared outputs (so-OPPRF), which enables oblivious evaluations of programmable PRF without revealing outputs to either party individually.

    \item \textbf{Modular fuzzy mapping.} We present a modular protocol design of fuzzy mapping, which is a core component of fuzzy PSI, from shared-output OPPRF and shared-input OPRF.

    \item \textbf{Efficient fuzzy PSI.} We propose efficient fuzzy PSI protocols for $L_{p \in [1, \infty]}$ distance metrics based on our fuzzy mapping, primarily leveraging lightweight symmetric-key operations, with linear complexity in the set size, the dimension, and the distance threshold.

    \item \textbf{Prefix optimizations.} We incorporate prefix techniques to optimize the overhead of fuzzy PSI protocols for large distance thresholds, which reduces the computational and communication complexity to logarithmic in the distance threshold.

    \item \textbf{Extensive evaluations.} We conduct extensive experiments under various parameter settings. Experimental results demonstrate that our protocols consistently and significantly outperform the state-of-the-art works across all settings, particularly, up to $80 \times$ faster computation and $19 \times$ lower communication.

\end{enumerate}

\subsection{Related Work}
\label{sec: Related Work}

\textbf{Fuzzy PSI for Hamming distance.} For a long time, research on fuzzy PSI primarily focused on the Hamming distance. Freedman et al.~\cite{freedman2004efficient} first introduced the problem of secure fuzzy matching and proposed a protocol based on additively homomorphic encryption (AHE). Subsequent works~\cite{chmielewski2008fuzzy,ye2009efficient,indyk2006polylogarithmic,uzun2021fuzzy,chakraborti2023distance,blass2025assumption} improve the protocol's complexity based on different techniques and assumptions.

\textbf{Fuzzy PSI for Minkowski distance.} More recently, a series of works have explored fuzzy PSI for $L_{p \in [1, \infty]}$ metrics. 
Garimella et al.~\cite{garimella2022structure} introduced the notion of structure-aware PSI, in which one party holds a structured dataset while the other holds an unstructured set of points. Their construction can be extended to fuzzy PSI for $L_\infty$ distance, but has a cost that scales with $\delta^d$. Chakraborti et al.~\cite{chakraborti2023distance} studied the problem for $L_1$ distance and presented a useful technique known as prefix representation.
Then, Garimella et al.~\cite{garimella2024computation} and Bui et al.~\cite{bui2025new} improved their previous work~\cite{garimella2022structure},
reducing the complexity to $(\log \delta)^d$ with prefix representation.

van Baarsen and Pu~\cite{van2024fuzzy} proposed the first fuzzy PSI protocol for general $L_\infty$ and $L_p$ distances based on the Decisional Diffie–Hellman (DDH) assumption.
Their construct supposed each element has a minimum distance ($2\delta$ or $4\delta$) from other elements.
They subsequently improved their design using lightweight primitives~\cite{van2025}, which significantly reduced the overall computational overhead. Similarly, they used prefix representation to achieve a complexity of $(\log \delta)^d$. 

Several follow-up works further improved the performance
but still incurred exponential complexity. Piske et al.~\cite{piske2025distance} constructed the fuzzy PSI protocol based on a newly introduced primitive, distance-aware OT. They followed the disjoint hash assumption~\cite{richardson2024fuzzy} where the whole input space is cut into grid cells and the spatial hashing scheme maps at most one element to every possible target grid cell. Their techniques rely on symmetric cryptographic primitives; however, their construction remains exponential in the dimension $d$. Moreover, their protocol operates over a relatively small domain and supports only two-dimensional inputs for the $L_2$ distance, which limits its practicality. Richardson et al.~\cite{richardson2024fuzzy} generalized the approach of~\cite{cho2016efficient} and built a fuzzy PSI protocol using generic two-party computation, resulting in poor concrete efficiency. Chongchitmate et al.~\cite{chongchitmate2024approximate} proposed a protocol that achieves nearly linear complexity in the set size. Nevertheless, their design relies on garbled circuits and multiple PSI invocations, which significantly limits its practical efficiency.

Building on~\cite{van2024fuzzy}, Gao et al.~\cite{gao2025efficient} proposed the first fuzzy PSI for $L_{p \in [1, \infty]}$ metrics, achieving linear complexity with respect to all parameters. They assumed that the set of both parties satisfies the disjoint projection assumption, where each element has disjoint projections in some dimension from all other elements. Later, Dang et al.~\cite{dang2025ccs} further improved this approach by adopting prefix representation to reduce the complexity to $\log \delta$. Nevertheless, both Gao et al.~\cite{gao2025efficient} and Dang et al.~\cite{dang2025ccs} heavily rely on the AHE, incurring considerable computational overhead. Recently, Zhang et al.~\cite{zhang2025fast} improved~\cite{gao2025efficient} by introducing a new framework for fuzzy PSI based on a stronger “$s$-separate” assumption.
Although~\cite {zhang2025fast} solely uses symmetric cryptographic primitives, it suffers from a relatively high exponential factor $\delta^s$ for both communication and computation overhead, especially for large $\delta$.

\textbf{Other Solutions.}
A closely related line of works study private record linkage (PRL)~\cite{de2011robust,durham2013composite,he2017composing,inan2010private,khurram2020sfour,wei2023cryptographically,stammler2022mainzelliste,adir2022privacy}, whose goal is to identify common entities across disparate datasets while protecting the privacy of non-matching records. Unlike our setting, these works typically do not rely on assumptions on the input distribution. Instead, many of them employ locality-preserving or locality-sensitive hashing techniques~\cite{zhao2014locality,leskovec2020mining} to project multi-dimensional records into lower-dimensional representations, often at the expense of matching accuracy.

Existing PRL protocols have been instantiated using a variety of privacy-enhancing techniques, including Bloom filters~\cite{bloom1970space}, differential privacy~\cite{dwork2006calibrating}, and secure multiparty computation (MPC). Among them, only a small number of works~\cite{khurram2020sfour,wei2023cryptographically,stammler2022mainzelliste,adir2022privacy} provide cryptographic security without additional leakage. In particular, Khurram and Kerschbaum~\cite{khurram2020sfour} proposed the first efficient cryptographically secure PRL protocol based on MPC. Their approach avoids the quadratic cost of all-pairs comparison by first securely sorting the records and then comparing only records that fall within the same comparison window. Nevertheless, the overall cost of their construction is dominated by the secure merge sorting procedure, resulting in an $O(n\log n)$ complexity, where $n$ denotes the total number of records. Wei et al.~\cite{wei2023cryptographically} improved this line of works in both efficiency and empirical accuracy by incorporating locality-sensitive hashing~\cite{broder1998min}. However, their protocol leaks the number of comparisons, and its overall complexity remains $O(n\log n)$. To achieve stronger privacy guarantees, Stammler et al.~\cite{stammler2022mainzelliste} abandoned locality-based blocking and instead compare all record pairs, incurring a quadratic $O(n^2)$ complexity. Adir et al.~\cite{adir2022privacy} combined locality-sensitive hashing~\cite{leskovec2020mining} with PSI to obtain linear complexity. However, replacing generic MPC-based circuit for distance computation with standard PSI significantly reduces the computational overhead, but this simplification comes at the cost of reduced matching accuracy. 

While PRL protocols address the trade-off among security, efficiency, and accuracy in more general application settings, our goal is to design a protocol for structured sets satisfying certain assumptions, which allows us to avoid introducing false positives or negatives. Accordingly, we treat PRL as a closely related series of works, but center our subsequent discussion on prior fuzzy PSI protocols that realize the same functionality under certain assumptions.

In this work, we focus on constructing concretely efficient fuzzy PSI protocols with linear complexity in both the set size and the dimension, while avoiding expensive public-key primitives such as AHE used in existing linear-complexity constructions~\cite{gao2025efficient,dang2025ccs}. Table~\ref{tab:complexities} compares the asymptotic complexities of our protocols with prior works for $L_{p\in[1,\infty]}$ distance. Among prior works under the same assumptions, we report only the protocols with the best asymptotic complexity.

\section{Technical  overview}
\label{sec: Technical  overview}

We begin by introducing the necessary notations.
The sender~$\SSS$ holds a set $Q = \{ q_j \}_{j \in [m]}$ of size~$m$, and the receiver~$\RRR$ holds a set $W = \{ w_i \}_{i \in [n]}$ of size~$n$, both defined over a $d$-dimensional space.
We use $q_{j, k}$ to denote the $k$-th coordinate of $q_j$.
For any $q_j \in Q$ and $w_i \in W$, we say that $q_j$ and $w_i$ are $\delta$-close if $\mathsf{dist}(q_j, w_i) \le \delta$, where $\delta$ denotes the distance threshold.

\subsection{Fuzzy PSI Paradigm from Fuzzy Mapping}

Recent works~\cite{van2024fuzzy, gao2025efficient, dang2025ccs, piske2025distance} show that fuzzy PSI protocols typically consist of two phases: \emph{coarse mapping} and \emph{refined filtering}. 
(1) In the coarse mapping phase, each element in both parties' sets is mapped to an identifier (ID). Any two elements that are within distance~$\delta$, one from the sender and one from the receiver, are guaranteed to receive the same ID. This phase may introduce false positives, but no false negatives. 
(2) In the refined filtering phase, for each candidate pair with the same ID, the parties further perform an exact distance check to eliminate the false positives generated during the coarse mapping phase.
The above idea is similar to hashing-based standard PSI protocols and reduces the number of necessary distance comparisons from $O(mn)$ to $O(n)$ or $O(m)$.

To realize the coarse mapping phase, Gao et al.~\cite{gao2025efficient} propose a generalized protocol, termed \emph{fuzzy mapping}, which takes sets $Q$ and $W$ as input and outputs identifier sets $\mathsf{ID}_{Q}$ and $\mathsf{ID}_{W}$ for the sender and receiver, respectively. 
The protocol guarantees that for any two elements $q_j \in Q$ and $w_i \in W$ satisfying $\mathsf{dist}(q_j, w_i) \leq \delta$, it holds that $\mathsf{ID}_{q_j} = \mathsf{ID}_{w_i}$. 
Subsequently, the refined filtering phase is carried out via \emph{fuzzy matching}, which takes as input every pair $(q_j, w_i)$ sharing the same ID and returns $q_j$ to the receiver if and only if $\mathsf{dist}(q_j, w_i) \leq \delta$.

In this work, we follow this fuzzy PSI paradigm, but introduce new techniques to realize the fuzzy mapping and fuzzy matching phases with linear complexity in the set size and dimension, primarily using lightweight symmetric-key operations\footnote{Our protocols invoke OT and VOLE, instantiated via pseudorandom correlation generators (PCGs)~\cite{boyle2019efficient, raghuraman2023expand}. In PCGs, the seed setup phase requires a small number of base OTs implemented with public-key operations, while the seed expansion phase relies on linear codes. Both phases are highly efficient in practice~\cite{van2025, piske2025distance}.}.

\subsection{OPPRF with Shared Outputs}
\label{subsec: OPPRF with shared outputs}

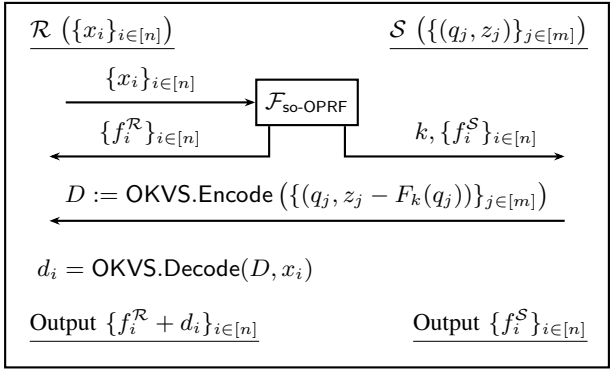
\begin{figure}[t]
    \centering
    \resizebox{0.46\textwidth}{!}{%
    \begin{tikzpicture}[
        node distance=3.5cm, %
        func_box/.style={
            draw, rectangle, thick, 
            minimum width=1.3cm, minimum height=0.6cm, 
            align=center, font=\small
        },
        myarrow/.style={
            ->, >={Stealth[scale=0.6]}, 
            draw=black, thick
        },
        lbl/.style={
            font=\small, align=center, text=black
        },
        container/.style={
            draw=none, thick, 
            rounded corners=0pt, 
            inner xsep=0cm,
            inner ysep=0cm
        }
    ]

    \node[func_box] (ssoprf) at (0, 0) {\Func[so\text{-}OPRF]};

    \draw[myarrow] ($(ssoprf.west)+(-2.5,0)$) -- (ssoprf.west) 
        node[pos=0.45, above, lbl] {$\{x_i\}_{i \in [n]}$};

    \draw[myarrow] ($(ssoprf.south)+(-0.5,0)$) |- ++(-2.9, -0.4) 
        node[pos=0.77, above, lbl] {$\{f_i^\RRR\}_{i \in [n]}$};

    \draw[myarrow] ($(ssoprf.south)+(0.5,0)$) |- ++(2.9, -0.4) 
        node[pos=0.80, above, lbl] {$k,\{f_i^\SSS\}_{i \in [n]}$};

    \draw[myarrow] ($(ssoprf.south)+(3.4,-1.25)$) -- ($(ssoprf.south)+(-3.4,-1.25)$)
    node[midway, above, lbl] {\small $D := \mathsf{OKVS.Encode}\left(\{(q_j, z_j-F_k(q_j))\}_{j \in [m]}\right)$};

    \node[lbl, anchor=south west] (calc_di) at ($(ssoprf.west)+(-3.0, -2.5)$) 
        {\small $d_i = \mathsf{ OKVS.Decode}(D, x_i)$};

    \node[lbl, anchor=south west] (out) at ($(ssoprf.west)+(-3.1, -3.3)$) 
    {\underline{Output $\{f_i^\RRR + d_i \}_{i\in [n]}$}};

    \node[lbl, anchor=south west] (out1) at ($(ssoprf.west)+(1.96, -3.3)$) 
    {\underline{Output $\{f_i^\SSS \}_{i\in [n]}$}};

    \path (ssoprf.east) ++(3.3,0) coordinate (aux_right);
    \path (ssoprf.west) ++(-3.3,0) coordinate (aux_left); %
    \coordinate (aux_bottom) at ($(out.south) + (0, -0.2cm)$); %
    \coordinate (aux_top) at ($(ssoprf.north) + (0, 1.0cm)$);
    
    \node[container, fit=(ssoprf) (aux_left) (aux_right) (aux_bottom)(aux_top), draw=black] (box) {};

    \node[font=\small] at ($(ssoprf.north) + (-2.7, 0.6)$) {\underline{\RRR $\left( \{x_i\}_{i\in[n]} \right)$}};
    
    \node[font=\small] at ($(ssoprf.north) + (2.4, 0.6)$) {\underline{\SSS $\left(\{(q_j,z_j) \}_{j\in[m]} \right)$}};

    \end{tikzpicture}
    }
    \caption{Construction of so-OPPRF from so-OPRF.}
    \label{fig: so-OPPRF workflow}
\end{figure}

An oblivious pseudorandom function (OPRF) allows a receiver to input a value~$x$ and obtain $\mathsf{PRF}(k, x)$, where the PRF key~$k$ is privately held by the sender. 
Programmable OPRF (OPPRF), introduced by Kolesnikov et al.~\cite{kolesnikov2017practical}, extends this functionality by allowing the sender to program the PRF on designated points: for each sender's chosen pair $(y_i, z_i)$, the functionality enforces $\mathsf{PRF}(k, y_i) = z_i$, while all non-programmed inputs receive pseudorandom outputs.

In this work, we introduce a variant of OPPRF, termed \emph{shared-output OPPRF} (so-OPPRF), which serves as a core building block in our fuzzy mapping and fuzzy PSI protocols. 
A so-OPPRF is identical to a standard OPPRF except that it outputs secret shares of $\mathsf{PRF}(k, x)$ to the two parties rather than revealing the PRF value to the receiver.

We construct an efficient so-OPPRF protocol by combining an oblivious key-value store (OKVS)~\cite{garimella2021oblivious} with a shared-output OPRF (so-OPRF)~\cite{van2024amortizing, alamati2024improved}, following the classical OPPRF paradigm~\cite{pinkas2019efficient, rindal2021vole}.
Concretely, the parties first run a so-OPRF, in which the sender inputs a PRF key~$k$ and the receiver inputs a set $X$, and they obtain secret shares $f_i^\SSS, f_i^\RRR$ of $\mathsf{PRF}(k, x_i)$ for each $x_i \in X$.
To program the PRF on the chosen points $\{(q_j, z_j)\}_{j \in [m]}$, the sender uses OKVS to encode these points.
OKVS \cite{garimella2021oblivious} consists of two algorithms: $\mathsf{Encode}$ takes as input a set of key-value pairs $\{(k_j, v_j)\}_{j \in [m]}$ and outputs an encoding $D$ while $\mathsf{Decode}$ takes as input a key $k$ and $D$, and outputs a value $v$.
The correctness ensures that $\mathsf{OKVS.Decode}(D, k_j)$ outputs $v_j$, and the obliviousness means that if the OKVS encodes random values, the encoding $D$ is independent of the encoded keys.
To this end, the sender computes an OKVS encoding $D$ on $\{(q_j, z_j - \mathsf{PRF}(k, q_j))\}_{j \in [m]}$
and sends $D$ to the receiver. 
Finally, the sender outputs $y_i^\SSS := f_i^\SSS$, while the receiver outputs $y_i^\RRR := f_i^\RRR + \mathsf{OKVS.Decode}(D, x_i).$
For each programmed point $x_i = q_j$, correctness holds immediately since $y_i^\SSS + y_i^\RRR
  = f_i^\SSS + f_i^\RRR + z_j - \mathsf{PRF}(k, q_j)
  = z_j.$
The pseudorandomness of unprogrammed points follows from the functionality of so-OPRF. 

The above steps are illustrated in Figure~\ref{fig: so-OPPRF workflow}. In the following Sections~\ref{subsec: Fuzzy mapping from so-OPPRF and si-OPRF} and~\ref{subsec: Fuzzy PSI from so-OPPRF}, we will show how to design fuzzy PSI from so-OPPRF.

\subsection{Fuzzy Mapping from so-OPPRF and si-OPRF}
\label{subsec: Fuzzy mapping from so-OPPRF and si-OPRF}

We present a modular framework of fuzzy mapping with the following two steps.
(1) Local mapping: the sender constructs a private local mapping $H_Q$ that maps the set $Q$'s elements to unique local IDs, while guaranteeing that for each $w_i \in W$, if there exists $q_j$ such that $\mathsf{dist}(q_j, w_i) \leq \delta$, $H_Q(w_i) = H_Q(q_j)$ holds. Similarly, the receiver generates a private local mapping $H_W$.
(2) Global mapping: two parties interactively compute the global IDs as $\mathsf{ID}_{q_j} := \mathsf{PRF}_k(H_Q(q_j) + H_W(q_j))$ and $\mathsf{ID}_{w_i} := \mathsf{PRF}_k(H_Q(w_i) + H_W(w_i))$ without revealing any information about $Q,W$.
Along with the correctness of local mapping, it ensures that if $\mathsf{dist}(q_j, w_i) \leq \delta$, it holds $\mathsf{ID}_{q_j} = \mathsf{ID}_{w_i}$.
It is worth noting that compared to the existing fuzzy mapping \cite{gao2025efficient, dang2025ccs}, our main contribution lies in the modular PRF-based global mapping design, which simplifies protocol comprehension and facilitates efficient protocol instantiation.
We elaborate on these aspects and the underlying technical methods in the remainder of this section.

\begin{figure}[t]
    \centering
    \resizebox{0.48\textwidth}{!}{%
    \begin{tikzpicture}[
        node distance=3.5cm, %
        func_box/.style={
            draw, rectangle, thick, 
            minimum width=1.3cm, minimum height=0.6cm, 
            align=center, font=\small
        },
        myarrow/.style={
            ->, >={Stealth[scale=0.6]}, 
            draw=black, thick
        },
        lbl/.style={
            font=\small, align=center, text=black
        },
        container/.style={
            draw=none, thick, 
            rounded corners=0pt, 
            inner xsep=0cm,
            inner ysep=0cm
        }
    ]

    \node[func_box] (mapping) at (0, 0) {Local Mapping};

    \node[func_box] (ssoprf) at (0, -2.0) {\Func[so\text{-}OPPRF]};

    \node[func_box] (sioprf) at (0, -4.0) {\Func[si\text{-}OPRF]};

    \node[
    single arrow,
    draw,
    line width=0.8pt,
    minimum width=0.08cm,
    minimum height=0.8cm,
    single arrow head extend=0.1cm,
    rotate=-90,
    scale=0.7
    ] at ($(mapping.south) + (0,-0.45)$) {};

    \node[
    single arrow,
    draw,
    line width=0.8pt,
    minimum width=0.08cm,
    minimum height=0.8cm,
    single arrow head extend=0.1cm,
    rotate=-90,
    scale=0.7
    ] at ($(ssoprf.south) + (0,-0.45)$) {};

    \node[lbl, anchor=north] (text1) at ($(ssoprf.north) + (0,0.55)$) 
    {Global Mapping};

    \node[lbl, anchor=north] (text2) at ($(sioprf.north) + (0,0.55)$) 
    {Randomize};

    \draw[myarrow] ($(mapping.east)+(2.25,0)$) -- (mapping.east) 
        node[pos=0.45, above, lbl] {$\{q_j\}_{j \in [m]}$};

    \draw[myarrow] ($(mapping.south)+(0.5,0)$) |- ++(2.9, -0.4) 
        node[pos=0.80, above, lbl] {$\{H_Q(q_j)\}_{j \in [m]}$};

    \draw[myarrow] ($(ssoprf.west)+(-2.5,0)$) -- (ssoprf.west) 
        node[pos=0.45, above, lbl] {$\{w_i\}_{i \in [n]}$};

    \draw[myarrow] ($(ssoprf.east)+(2.5,0)$) -- (ssoprf.east) 
        node[pos=0.45, above, lbl] {$H_Q(\cdot)$};

    \draw[myarrow] ($(ssoprf.south)+(-0.5,0)$) |- ++(-2.77, -0.4) 
        node[pos=0.80, above, lbl] {$\{H_Q^\RRR(w_i)\}_{i \in [n]}$};

    \draw[myarrow] ($(ssoprf.south)+(0.5,0)$) |- ++(2.77, -0.4) 
        node[pos=0.80, above, lbl] {$\{H_Q^\SSS(w_i)\}_{i \in [n]}$};

    \draw[myarrow] ($(sioprf.west)+(-2.5,0)$) -- (sioprf.west) 
node[pos=0.35, above, lbl]
{\scriptsize $i\in[n]$ \\ $\{H_Q^\RRR(w_i)\!+\!H_W(w_i)\}$};

    \draw[myarrow] ($(sioprf.east)+(2.5,0)$) -- (sioprf.east) 
        node[pos=0.45, above, lbl] {$\{H_Q^\SSS(w_i) \}_{i\in[n]}$};

    \draw[myarrow] ($(sioprf.south)+(-0.5,0)$) |- ++(-2.73, -0.4) 
        node[pos=0.80, above, lbl] {$\{\textsf{ID}_{w_i}\}_{i \in [n]}$};

    \node[lbl, anchor=south west] (out) at ($(ssoprf.west)+(-3.1, -3.55)$) 
    {\underline{Output $\{\textsf{ID}_{w_i}\}_{i \in [n]}$}};

    \path (ssoprf.east) ++(3.3,0) coordinate (aux_right);
    \path (ssoprf.west) ++(-3.3,0) coordinate (aux_left); %
    \coordinate (aux_bottom) at ($(out.south) + (0, -0.2cm)$); %
    \coordinate (aux_top) at ($(mapping.north) + (0, 1.0cm)$);
    
    \node[container, fit= (mapping) (ssoprf) (aux_left) (aux_right) (aux_bottom)(aux_top), draw=black] (box) {};

    \node[font=\small] at ($(mapping.north) + (-2.8, 0.6)$) {\underline{\RRR $\left( \{w_i\}_{i\in[n]} \right)$}};
    
    \node[font=\small] at ($(mapping.north) + (2.8, 0.6)$) {\underline{\SSS $\left(\{q_j \}_{j\in[m]} \right)$}};

    \end{tikzpicture}
    }
    \caption{Construction of fuzzy mapping from so-OPPRF and si-OPRF. The sender and receiver can switch roles and repeat the process symmetrically.}
    \label{fig: fpsi workflow}
\end{figure}
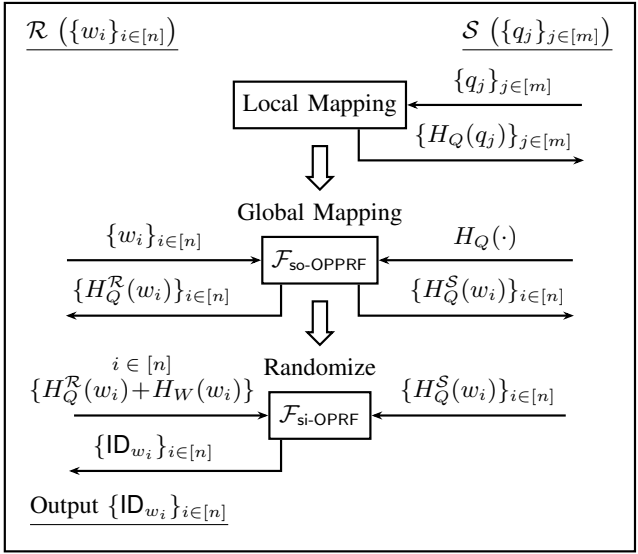

To instantiate the above framework, we present an efficient fuzzy mapping protocol for $L_\infty$ distance. 
The following description focuses on the receiver's protocol due to the symmetric design.
Specifically, in the first local mapping step, same as the prior works \cite{gao2025efficient, zhang2025fast, dang2025ccs}, for each $q_{j, k}$, the sender locally assigns a random value $r_{j, k}$ to the interval $[q_{j, k} - \delta, q_{j, k} + \delta]$ centered on $q_{j, k}$, formally represented as a set of the key-value pairs $L := \{ (k \| (q_{j, k} + t), r_{j, k}) \}_{j \in [m], k \in [d], t \in [-\delta, \delta]}$. We will explain how to handle overlaps later.
The sender then defines the local mapping $H_Q$ such that $H_Q(q_j+\xi) := \sum_{k \in [d]} r_{j, k}$ for each $q_j \in Q$ and $\xi \in [-\delta, \delta]^d$.
Then, in the global mapping step, to obtain $H_Q(w_i)$ without revealing any information to each other, the receiver obliviously retrieves values $r_{w_i, k}$ from the sender's key-value list $L$ with the inputs $k \| w_{i, k}$ for each dimension $k$.
It needs to ensure that for each $k \in [d]$, if $\abs{w_{i, k} - q_{j, k}} \leq \delta$, then $r_{w_i, k} = r_{j, k}$ holds and hence $H_Q(w_i) = \sum_{k \in [d]} r_{w_i, k} = H_Q(q_j)$; otherwise, $r_{w_i, k}$ is a random value.

We observe that the above functionality can be realized by invoking the OPPRF protocol, where the sender inputs the programmed points $L$, and the receiver inputs queries $k \| w_{i, k}$ and learns either the designated $r_{j, k}$ or pseudorandom values.
However, this will leak to the receiver additional information about partial matches on individual dimensions.
That is, since the sender programs all points in the interval $[q_{j, k} - \delta, q_{j, k} + \delta]$ to the same $r_{j, k}$, the receiver learns the same OPPRF evaluation on two different $w_{i,k}, w_{i^\prime, k} \in [q_{j, k} - \delta, q_{j, k} + \delta]$, leaking the range of some $q_{j, k}$.
Existing solutions \cite{gao2025efficient, dang2025ccs} employ AHE to avoid this leakage, but result in prohibitively large overhead.

\noindent\textbf{Applying so-OPPRF.} We adopt our so-OPPRF protocol to prevent the information leakage to the receiver while achieving better efficiency than AHE-based solutions.
Specifically, with the same inputs $L$ and $\{k \| w_{i, k} \}_{k \in [d]}$ as above, so-OPPRF returns the secret shares $\{ r^{\SSS}_{w_i,k} \}_{k\in [d]}$ to the sender and $\{ r^{\RRR}_{w_i,k} \}_{k\in [d]}$ to the receiver.
Then, the sender sets $H^\SSS_Q(w_i) := \sum_{k\in[d]} r^{\SSS}_{w_i,k}$ and the receiver sets $H^\RRR_Q(w_i) := \sum_{k\in[d]} r^{\RRR}_{w_i,k}$, both of which are the secret shares of $H_Q(w_i)$.

However, $H_Q(w_i)$ can not be directly revealed to the receiver.
The reason is that in the local mapping~\cite{gao2025efficient}, when there are overlapped intervals, later random values will overwrite earlier values. This enlarges the valid radius-$\delta$ interval of $q_{j,k}$, thereby introducing false positives.
This may result in two elements $w_{i_1}, w_{i_2}$ satisfying $H_Q(w_{i_1}) = H_Q(w_{i_2})$ even if $\mathsf{dist}(w_{i_1}, w_{i_2}) > \delta$, which leaks the sender's information, i.e., there is a sender’s element nearby $w_{i_1}, w_{i_2}$.
Existing protocols \cite{gao2025efficient, dang2025ccs} employ a symmetric execution\footnote{\label{fn: revise}The method \cite{dang2025ccs} lacks this symmetric execution and causes the above privacy leakage issue, which has been confirmed by the authors via our private communication. Table \ref{tab:complexities} gives the corrected complexity with the symmetric execution.} to additionally compute $H_W(w_i)$ and derive $\mathsf{ID}_{w_i}$ from $H_Q(w_i) + H_W(w_i)$ with a customized combination of expensive AHE and Diffie–Hellman key-exchange.

\noindent\textbf{Applying si-OPRF.} To address this leakage, we introduce a new building block, called OPRF with secret-shared inputs (si-OPRF), and realize it with MPC-friendly PRF \cite{alamati2024improved}. 
Different from so-OPRF, in si-OPRF, the PRF key $k$ and the PRF input $x$ are secret shared between two parties. 
After evaluation, it returns the plain PRF output $\mathsf{PRF}_k(x)$ to the receiver.
To obtain the final IDs $\mathsf{ID}_{w_i}$, the sender and the receiver input secret shares $H^\SSS_Q(w_i)$ and $H^\RRR_Q(w_i) + H_W(w_i)$ and secret shared $k$ into si-OPRF. Then the receiver learns $\mathsf{ID}_{w_i} := \mathsf{PRF}_k(H_Q(w_i)+H_W(w_i))$.
This ensures that even for $H_Q(w_{i_1}) = H_Q(w_{i_2})$, if local mapping ensures that $H_W(w_{i_1}) \neq H_W(w_{i_2})$, the receiver learns different $\mathsf{ID}_{w_{i_1}}$ and $ \mathsf{ID}_{w_{i_2}}$.
Similarly, the sender learns $\mathsf{ID}_{q_j} := \mathsf{PRF}_k(H_Q(q_j)+H_W(q_j))$, where the same secret-shared PRF key $k$ is used by both parties in two si-OPRF invocations.
This ensures that $\mathsf{ID}_{q_j} = \mathsf{ID}_{w_i}$ for $\mathsf{dist}(q_{j}, w_{i}) \leq \delta$ due to $H_Q(q_j) = H_Q(w_i)$ and $H_W(q_j) = H_W(w_i)$.
The above processes are illustrated in Figure~\ref{fig: fpsi workflow}.

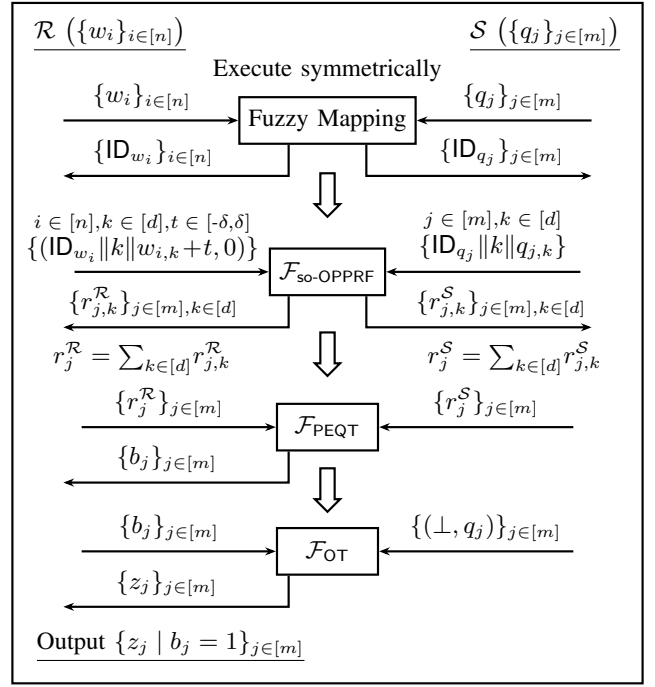
\begin{figure}[t]
    \centering
    \resizebox{0.48\textwidth}{!}{%
    \begin{tikzpicture}[
        node distance=3.5cm, %
        func_box/.style={
            draw, rectangle, thick, 
            minimum width=1.3cm, minimum height=0.6cm, 
            align=center, font=\small
        },
        myarrow/.style={
            ->, >={Stealth[scale=0.6]}, 
            draw=black, thick
        },
        lbl/.style={
            font=\small, align=center, text=black
        },
        container/.style={
            draw=none, thick, 
            rounded corners=0pt, 
            inner xsep=0cm,
            inner ysep=0cm
        }
    ]

    \node[func_box] (mapping) at (0, -0.25) {Fuzzy Mapping};

    \node[func_box] (ssoprf) at (0, -2.2) {\Func[so\text{-}OPPRF]};

    \node[func_box] (sioprf) at (0, -4.2) {\Func[PEQT]};

    \node[func_box] (ot) at (0, -5.8) {\Func[OT]};

    \node[
    single arrow,
    draw,
    line width=0.8pt,
    minimum width=0.08cm,
    minimum height=0.8cm,
    single arrow head extend=0.1cm,
    rotate=-90,
    scale=0.7
    ] at ($(mapping.south) + (0,-0.60)$) {};

    \node[
    single arrow,
    draw,
    line width=0.8pt,
    minimum width=0.08cm,
    minimum height=0.8cm,
    single arrow head extend=0.1cm,
    rotate=-90,
    scale=0.7
    ] at ($(ssoprf.south) + (0,-0.70)$) {};

    \node[
    single arrow,
    draw,
    line width=0.8pt,
    minimum width=0.08cm,
    minimum height=0.8cm,
    single arrow head extend=0.1cm,
    rotate=-90,
    scale=0.7
    ] at ($(sioprf.south) + (0,-0.45)$) {};

    \node[lbl, anchor=north] (text0) at ($(mapping.north) + (0,0.6)$) 
    {Execute symmetrically};

    \draw[myarrow] ($(mapping.east)+(2.25,0)$) -- (mapping.east) 
        node[pos=0.45, above, lbl] {$\{q_j\}_{j \in [m]}$};

    \draw[myarrow] ($(mapping.west)+(-2.25,0)$) -- (mapping.west) 
        node[pos=0.45, above, lbl] {$\{w_i\}_{i \in [n]}$};

    \draw[myarrow] ($(mapping.south)+(-0.5,0)$) |- ++(-2.9, -0.4) 
        node[pos=0.80, above, lbl] {$\{\textsf{ID}_{w_i}\}_{i \in [n]}$};

    \draw[myarrow] ($(mapping.south)+(0.5,0)$) |- ++(2.9, -0.4) 
        node[pos=0.80, above, lbl] {$\{\textsf{ID}_{q_j}\}_{j \in [m]}$};

    \draw[myarrow] ($(ssoprf.west)+(-2.5,0)$) -- (ssoprf.west) 
        node[pos=0.35, above, lbl] {\scriptsize $i \in [n],\!k
    \in [d], \!t\in [\text{-}\delta,\!\delta]$ \\ $\left\{(\textsf{ID}_{w_i}\Vert k \Vert w_{i,k}\!+\!t,0)\right\}$};

    \draw[myarrow] ($(ssoprf.east)+(2.5,0)$) -- (ssoprf.east) 
        node[pos=0.45, above, lbl] {\scriptsize $j \in [m],\!k
    \in [d]$ \\ $\{\textsf{ID}_{q_j}\Vert k \Vert q_{j,k} \}$};

    \draw[myarrow] ($(ssoprf.south)+(-0.5,0)$) |- ++(-2.9, -0.4) 
        node[pos=0.80, above, lbl] {$\{r_{j,k}^\RRR\}_{j \in [m],k
    \in [d]}$};

    \draw[myarrow] ($(ssoprf.south)+(0.5,0)$) |- ++(2.9, -0.4) 
        node[pos=0.80, above, lbl] {$\{r_{j,k}^\SSS\}_{j \in [m],k
    \in [d]}$};

    \node[lbl, anchor=south west] (calc_di0) at ($(ssoprf.west)+(-2.9, -1.45)$) 
        {$r_{j}^\RRR = {\sum}_{k\in[d]} r_{j,k}^\RRR$};

    \node[lbl, anchor=south east] (calc_di1) at ($(ssoprf.east)+(2.9, -1.45)$) 
        {$r_{j}^\SSS = {\sum}_{k\in[d]} r_{j,k}^\SSS$};

    \draw[myarrow] ($(sioprf.west)+(-2.5,0)$) -- (sioprf.west) 
node[pos=0.45, above, lbl]
{$\{r_j^\RRR\}_{j\in[m]}$};

    \draw[myarrow] ($(sioprf.east)+(2.5,0)$) -- (sioprf.east) 
        node[pos=0.45, above, lbl] {$\{r_j^\SSS\}_{j\in[m]}$};

    \draw[myarrow] ($(sioprf.south)+(-0.5,0)$) |- ++(-2.9, -0.4) 
        node[pos=0.77, above, lbl] {$\{b_j\}_{j\in[m]}$};

    \draw[myarrow] ($(ot.west)+(-2.5,0)$) -- (ot.west) 
node[pos=0.45, above, lbl]
{$\{b_j\}_{j\in[m]}$};

    \draw[myarrow] ($(ot.east)+(2.5,0)$) -- (ot.east) 
        node[pos=0.45, above, lbl] {$\{(\bot,q_j)\}_{j\in[m]}$};

    \draw[myarrow] ($(ot.south)+(-0.5,0)$) |- ++(-2.9, -0.4) 
        node[pos=0.77, above, lbl] {$\{z_j\}_{j\in [m]}$};

    \node[lbl, anchor=south west] (out) at ($(sioprf.west)+(-3.2, -3.15)$) 
    {\underline{Output $\{ z_j \;\vert \; b_j=1\}_{j\in [m]}$}};

    \path (ssoprf.east) ++(3.3,0) coordinate (aux_right);
    \path (ssoprf.west) ++(-3.3,0) coordinate (aux_left); %
    \coordinate (aux_bottom) at ($(out.south) + (0, -0.15cm)$); %
    \coordinate (aux_top) at ($(mapping.north) + (0, 1.2cm)$);
    
    \node[container, fit= (mapping) (ssoprf) (aux_left) (aux_right) (aux_bottom)(aux_top), draw=black] (box) {};

    \node[font=\small] at ($(mapping.north) + (-2.8, 0.8)$) {\underline{\RRR $\left( \{w_i\}_{i\in[n]} \right)$}};
    
    \node[font=\small] at ($(mapping.north) + (2.8, 0.8)$) {\underline{\SSS $\left(\{q_j \}_{j\in[m]} \right)$}};

    \end{tikzpicture}
    }
    \caption{Construction of fuzzy PSI from so-OPPRF.}
    \label{fig: fuzzy psi workflow}
\end{figure}

\subsection{Fuzzy PSI from so-OPPRF}
\label{subsec: Fuzzy PSI from so-OPPRF}

The fuzzy mapping protocol above guarantees that all $\delta$-close elements between the sender and the receiver are assigned the same IDs, but it may introduce false positives. We further utilize our so-OPPRF protocol to perform a refined filtering that excludes elements that are not truly $\delta$-close, thereby achieving fuzzy PSI for both the $L_\infty$ and $L_{p}$ distances with $p \in [1, \infty)$.

We here focus on the protocol for $L_\infty$ distance.
As shown in Figure~\ref{fig: fuzzy psi workflow}, two parties first invoke fuzzy mapping procedure to get the IDs, and then execute the so-OPPRF protocol, where the receiver inputs $\{\left( \mathsf{ID}_{w_i} \Vert k\Vert (w_{i,k}+t), 0\right)\}_{t\in [-\delta,\delta], i \in [n],k\in [d]}$ and the sender inputs $\{ \mathsf{ID}_{q_j} \Vert k \Vert q_{j,k} \}_{j \in [m],k \in [d]}$.
The sender receives $\{ r^{\SSS}_{j,k} \}_{j\in [m],k\in [d]}$ and the receiver receives $\{ r^{\RRR}_{j,k} \}_{j\in[m],k\in [d]}$.
That means if $\mathsf{ID}_{q_j} = \mathsf{ID}_{w_i}$ and the distance between $q_{j,k}$ and $w_{i,k}$ is within $\delta$, both parties will obtain the secret shares of 0 in dimension $k$; otherwise, they hold random values.
Then, for \( j \in [m] \), the sender computes $r^{\SSS}_{j} :={\sum}_{k\in [d]} r^{\SSS}_{j,k}$, and the receiver computes $r^{\RRR}_{j} := {\sum}_{k\in [d]} r^{\RRR}_{j,k}$.
For each $\delta$-close element pair with the same ID, so-OPPRF guarantees that $r^{\SSS}_{j}$ and $r^{\RRR}_{j}$ are secret shares of the designated value 0.
In contrast, if $\mathsf{ID}_{q_j}$ is different from all IDs of $W$ or there exists $w_i$ with the same ID but the distance $\abs{w_{i,k} - q_{j,k}} > \delta$ for some dimension $k$, then $r^{\SSS}_{j}$ and $r^{\RRR}_{j}$ will be secret shares of a random value according to the randomness property of so-OPPRF.
After that, similar to prior works, both parties can employ a cheap Private Equality Test (PEQT) protocol to verify whether they hold secret shares of 0. Finally, the two parties invoke an Oblivious Transfer (OT) protocol, where the receiver selects the sender's elements based on outcomes of the PEQT protocol.

\subsection{Fuzzy PSI with prefix optimizations}
\label{subsec: Fuzzy PSI with prefix optimizations}

We further optimize our fuzzy PSI protocols using prefix techniques \cite{garimella2024computation,van2025,dang2025ccs,bui2025new} for large distance threshold $\delta$.
This modification reduces the communication and computation complexity from $\bigo{\delta}$ to $\bigo{\log\delta}$.
The high-level idea to incorporate prefix techniques in our framework is that the sender in the so-OPPRF protocol does not need to program the entire interval $[q_{j,k}-\delta, q_{j,k}+\delta]$ of size $2\delta + 1$, but only $\bigo{\log\delta}$ prefixes that together cover the interval. We present customized protocols to instantiate this optimization. 
In contrast to the state-of-the-art protocol \cite{dang2025ccs} that employs expensive AHE, our protocols still use our so-OPPRF and the newly introduced equality-conditional selection protocol, which significantly improve overall efficiency.
Please refer to Section \ref{sec: Fuzzy PSI with Prefix Optimization} for more details.

\section{Preliminary}
\label{sec: Preliminary}

\subsection{Notation}
We use \(\kappa\) and \(\lambda\) to denote the computational and statistical security parameters, respectively. 
Let $\mathsf{negl}(x)$ be a negligible function in $x$ if it vanishes faster than the inverse of any polynomial in $x$.
We use $[a, b]$ to denote the set $\{a, \ldots, b\}$ and $[a]$ to denote the set $\{1, \ldots, a\}$.
For a set \(S\), \(|S|\) denotes the cardinality of $S$.
By $r \leftarrow S$, we denote that $r$ is sampled from the set $S$ uniformly at random.
We use $\textbf{1}\{\mathsf{event}\}$ to denote an indicator function, which equals 1 if the $\mathsf{event}$ occurs and 0 otherwise. 
All protocols in this work are secure in the semi-honest model and the definition is deferred to Appendix~\ref{sec: Threat Model}.

\subsection{Functionality of Fuzzy PSI}
We formally define the ideal functionality for fuzzy PSI in Figure~\ref{Func:FPSI}.
Same as the state-of-the-art linear-complexity fuzzy PSI works \cite{gao2025efficient, dang2025ccs}, we utilize the disjoint projection assumption for the sets of both sender and receiver.
That is, each element maintains a distance of more than $2\delta$ on at least one dimension from the set's other elements.
We define it formally as follows:
\begin{definition}[Disjoint projection]
\label{def:assumption}
    A set $W \in \UU^{n \times d}$ satisfies the $\delta$-disjoint projection assumption, if for any $w_i \in W$, there exists $k \in [d]$ such that for any $w_j \in W$ for $j \neq i$ it holds $[w_{i,k} - \delta, w_{i,k} + \delta] \intersection [w_{j,k} - \delta, w_{j,k} + \delta] = \emptyset$.
\end{definition}
As shown in the work~\cite{van2024fuzzy}, if the set's elements are uniformly distributed, the set satisfies the disjoint projection assumption defined in Definition~\ref{def:assumption} with probability $1-\mathsf{negl}(d)$.

\begin{figure}[h]
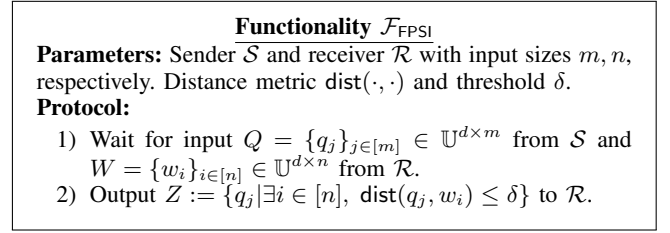

\begin{nffunc}{\Func[FPSI]}

\noindent \textbf{Parameters:} Sender \SSS and receiver \RRR with input sizes $m, n$, respectively. Distance metric $\mathsf{dist}(\cdot,\cdot)$ and threshold $\delta$.

\noindent \textbf{Protocol:}

\begin{enumerate}

    \item Wait for input $Q = \{q_j\}_{j \in [m]} \in \UU^{d\times m}$ from \SSS and $W = \{w_i\}_{i \in [n]} \in \UU^{d\times n}$ from \RRR.

    \item Output $Z := \{q_j \vert \exists i\in[n], \;\mathsf{dist}(q_j, w_i) \le \delta \}$ to \RRR.

\end{enumerate}

\end{nffunc}
\vspace{0.5em}
\caption{Functionality of fuzzy PSI.}
\label{Func:FPSI}
\end{figure}

\begin{figure}[t]
\begin{nffunc}{\Func[so\text{-}OPRF]}

\noindent \textbf{Parameters:} Sender \SSS and receiver \RRR. 
PRF $F : \KK \times \UU \to \FF.$

\noindent \textbf{Protocol:}

\begin{enumerate}

    \item  Wait for input $X=\{ x_1,\ldots,x_n\} \subseteq \UU$  from \RRR.

    \item  Wait for input $k \in \KK$ from \SSS.

    \item  For $i \in [n]$, compute $y_i := F_k(x_i)$ and sample $y_i^\SSS, y_i^\RRR \from \FF$ such that $y_i^\SSS + y_i^\RRR = y_i \in \FF$.

    \item Output $\{y_i^\SSS\}_{i \in [n]}$ to \SSS and $\{y_i^\RRR\}_{i \in [n]}$ to \RRR. 
    
\end{enumerate}

\end{nffunc}
\vspace{0.5em}
\caption{Functionality of oblivious PRF with secret-shared outputs.}
\label{Func:so-oprf}
\end{figure}

\begin{figure}[h]
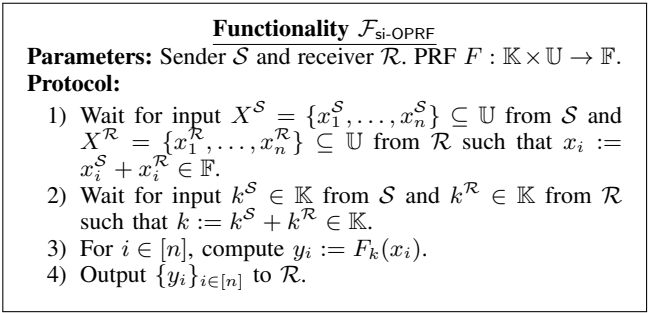

\begin{nffunc}{\Func[si\text{-}OPRF]}

\noindent \textbf{Parameters:} Sender \SSS and receiver \RRR. PRF $F : \KK \times \UU \to \FF.$

\noindent \textbf{Protocol:}

\begin{enumerate}

    \item  Wait for input $X^\SSS=\{ x_1^\SSS,\ldots,x_n^\SSS\} \subseteq \UU$  from \SSS and $X^\RRR=\{ x_1^\RRR,\ldots,x_n^\RRR\} \subseteq \UU$ from \RRR such that $x_i := x_i^\SSS + x_i^\RRR \in \FF$.

    \item  Wait for input $k^\SSS \in \KK$ from \SSS and $k^\RRR \in \KK$ from \RRR such that $k := k^\SSS + k^\RRR \in \KK$.

    \item  For $i \in [n]$, compute $y_i := F_k(x_i)$.

    \item Output $\{y_i\}_{i \in [n]}$ to \RRR.
    
\end{enumerate}

\end{nffunc}
\vspace{0.5em}
\caption{Functionality of oblivious PRF with secret-shared inputs.}
\label{Func:si-oprf}
\end{figure}

\subsection{Oblivious Key-Value Store}

An oblivious key-value store (OKVS) \cite{garimella2021oblivious} is a data structure consisting of two algorithms:
$\mathsf{Encode}$ takes as input a set of key-value pairs and outputs a data structure, and $\mathsf{Decode}$ takes as input a key and the data structure and outputs a value.
The obliviousness implies that if the OKVS encodes random values, the data structure is independent of the encoded key set.

\begin{definition}[Oblivious Key-Value Store \cite{garimella2021oblivious}]
\label{def: okvs}
An oblivious key-value store (OKVS) is parameterized by a key space \KKK, a value space \VVV, and the statistical security parameter $\lambda$, and consists of two algorithms:
\begin{itemize}
    \item $\mathsf{Encode}$: on input a set of key-value pairs $L \in (\KKK \times \VVV)^n$, outputs a vector $D \in \VVV^m$ or a failure indicator $\bot$ with probability bounded by $1/2^\lambda$.
    \item $\mathsf{Decode}$: on input a vector $D \in \VVV^m$ and a key $k \in \KKK$, outputs a value $v \in \VVV$.
\end{itemize}
\end{definition}

\noindent \textbf{Correctness}: For all $L \in (\KKK \times \VVV)^n$ with distinct keys for which $D \from \mathsf{Encode}(L)$ and $D \neq \bot$, it holds that $\forall (k, v) \in L$, $\mathsf{Decode}(D, k) = v$.

\noindent \textbf{Obliviousness}: For any distinct $\{k_1, \ldots, k_n\} \in \KKK^n$ and $\{k^\prime_1, \ldots, k^\prime_n\} \in \KKK^n$, $\mathsf{Encode}$ does not output $\bot$ on $\{k_1, \ldots, k_n\}$ and $\{k^\prime_1, \ldots, k^\prime_n\}$, and then the following distributions are statistically indistinguishable
\begin{equation*}
\begin{aligned}
& \{\mathsf{Encode}(\{(k_1, v_1), \ldots, (k_n, v_n)\}) \mid v_i \from \VVV, i \in [n] \} \approx_s \\
 & \{\mathsf{Encode}(\{(k^\prime_1, v_1), \ldots, (k^\prime_n, v_n)\}) \mid v_i \from \VVV, i \in [n]\}.
\end{aligned}
\end{equation*}

\noindent \textbf{Double obliviousness}: For any distinct $\{k_1, \ldots, k_n\} \in \KKK^n$ such that $\mathsf{Encode}$ does not output $\bot$ on $\{k_1, \ldots, k_n\}$, then $\{\mathsf{Encode}(\{(k_1, v_1), \ldots, (k_n, v_n)\}) \mid v_i \from \VVV, i \in [n]\}$ is statistically indistinguishable from uniform distribution over $\VVV^m$.

\noindent \textbf{Independence}: For any $L := \{(k_i, v_i)_{i \in [n]}\} \in (\KKK \times \VVV)^n$ with distinct keys such that $D \from \mathsf{Encode}(L)$ and $D \neq \bot$, it holds that for any $k \notin \{k_i\}_{i \in [n]}$, $\mathsf{Decode}(D, k)$ is statistically indistinguishable from uniform distribution over $\VVV$.

\subsection{Secret shared OPRF}

We introduce two functionalities of  OPRF with secret shared output (so-OPRF) and OPRF with secret shared key and input (si-OPRF) in Figure~\ref{Func:so-oprf} and Figure~\ref{Func:si-oprf}, respectively.
We instantiate the above two variants of OPRF using MPC-friendly alternating-moduli (weak) PRF, which was first proposed by Boneh et al. \cite{boneh2018exploring} and later optimized in
\cite{dinur2021mpc, albrecht2024crypto, alamati2024improved}.
At a high level, this construction takes as input a key $k$ and a value $x$, and multiplies them by various matrices modulo two different primes, e.g., 2 and 3. 
More concretely, the PRF construction in \cite{alamati2024improved} is defined as
$
F(k,x) := \mathbf{B} \cdot_2 \left(\mathbf{A} \cdot_3 \left[k \circ_2 \left(\mathbf{G} \cdot_2 \left[x || 1\right]\right)\right]\right)
$,
where $\mathbf{G} \in \mathbb{F}_2^{n \times (d+1)}$, $\mathbf{A} \in \mathbb{F}_3^{m \times n}$, $\mathbf{B} \in \mathbb{F}_2^{t \times m}$ are uniformly distributed, and $\cdot_p$, $\circ_p$ are multiplication and component-wise multiplication modulo $p$. In particular, we make use of the parameterization $d = \kappa, n = 4\kappa, m = 2\kappa, t = \kappa$ which implies $\mathbf{G}$ is an expanding matrix and $\mathbf{B}$ is a compressing matrix. 
The security of this construction relies on the assumption~\cite{boneh2018exploring} that linear operations over two different moduli result in a highly non-linear and unpredictable function when viewed over $\mathbb{F}_2$ or $\mathbb{F}_3$.

\subsection{Functionalities of Building Blocks}

We present the functionalities of the private equality/interval test in Figure~\ref{Func:Interval}.
Chakraborti et al.~\cite{chakraborti2023distance} introduced an efficient private interval test protocol based on prefix representation. Both the communication and computation complexities are $\bigo{\log \delta}$.
Besides, we introduce the functionalities \Func[MUX] and \Func[ssPEQT], which are deferred to Appendix~\ref{sec: mpc func}.

\begin{figure}[htbp]
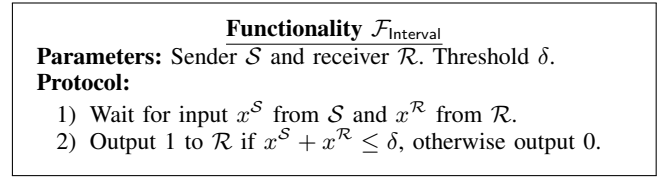

\begin{nffunc}{\Func[Interval]}

\noindent \textbf{Parameters:} Sender \SSS and receiver \RRR. Threshold $\delta$.

\noindent \textbf{Protocol:}

\begin{enumerate}

    \item Wait for input $x^{\SSS}$ from \SSS and $x^{\RRR}$ from \RRR.

    \item Output 1 to \RRR if $x^{\SSS}+x^{\RRR} \le \delta$, otherwise output 0.

\end{enumerate}

\end{nffunc}
\vspace{0.5em}
\caption{Functionality of private interval test.}
\label{Func:Interval}
\end{figure}

\begin{figure}[t]
\begin{nffunc}{\Func[so\text{-}OPPRF]}

\noindent \textbf{Parameters:} Sender \SSS and receiver \RRR with input sizes $m, n$, respectively. 

\noindent \textbf{Protocol:}

\begin{enumerate}

    \item Wait for input $L = \{(q_{j}, z_j)\}_{j \in [m]} \subseteq \UU \times \FF$ from \SSS and $X = \{x_i\}_{i \in [n]} \subseteq \UU$ from \RRR.

    \item Sample a random function $F^\prime : \UU \to \FF$ such that $F^\prime(q) = z$ for $(q, z) \in L$. 
    
    \item For $i \in [n]$, compute $y_i := F^\prime(x_i)$ and sample $y_i^\SSS, y_i^\RRR \from \FF$ such that $y_i^\SSS + y_i^\RRR = y_i \in \FF$.

    \item Output $\{y_i^\SSS\}_{i \in [n]}$ and $\OOO^{F^\prime}$ to \SSS and $\{y_i^\RRR\}_{i \in [n]}$ to \RRR.

\end{enumerate}

\end{nffunc}
\vspace{0.5em}
\caption{Functionality of oblivious programmable PRF with secret-shared outputs.}
\label{Func:so-opprf}
\end{figure}

\begin{figure}[t]
\begin{nfprot}{\Prot[so\text{-}OPPRF]}

\noindent \textbf{Parameters:} Sender \SSS and receiver \RRR with input sizes $m, n$, respectively.  
PRF $F : \KK \times \UU \to \FF.$

\noindent \textbf{Input:} \SSS inputs $L = \{(q_{j}, z_j)\}_{j \in [m]} \subseteq \UU \times \FF$. \RRR inputs $X = \{x_i\}_{i \in [n]} \subseteq \UU$. 

\noindent \textbf{Protocol:}

\begin{enumerate}

    \item \SSS samples $k \from \KK$. \SSS and \RRR invoke  functionality \Func[so\text{-}OPRF], where \SSS inputs $k$ and $R$ inputs $X = \{x_i\}_{i \in [n]}$. \SSS receives $\{f_i^\SSS\}_{i \in [n]}$, and \RRR receives $\{f_i^\RRR\}_{i \in [n]}$, where $f_i := F_k(x_i) \in \FF$.

    \item \SSS computes PRF values $F_k(q_j)$ for $j \in [m]$ and computes the OKVS encoding $D := \mathsf{OKVS.Encode}(L^\prime)$, where $L^\prime := \{ (q_j, z_j - F_k(q_j))\}_{j \in [m]}$.

    \item \SSS sends $D$ to \RRR and \RRR computes $d_i := \mathsf{OKVS.Decode}(D, x_i)$ for $i \in [n]$.

    \item  \SSS defines the function $F^\prime(q) := F_k(q) + \mathsf{OKVS.Decode}(D, q)$.

    \item \SSS outputs $\{y_i^\SSS := f_i^\SSS\}_{i \in [n]}$ and $\OOO^{F^\prime}$, and \RRR outputs $\{y_i^\RRR := f_i^\RRR + d_i\}_{i \in [n]}$.

\end{enumerate}

\end{nfprot}
\vspace{0.5em}
\caption{Protocol of oblivious programmable PRF with secret-shared outputs.}
\label{Prot:so-opprf}
\end{figure}

\subsection{Prefix Representation}

Prefix representation, first proposed by Chakraborti et al.~\cite{chakraborti2023distance} and further studied and formalized in~\cite{garimella2024computation, dang2025ccs, van2025,bui2025new}, is used to improve the efficiency of fuzzy PSI. 
Let $x = x_{\ell} x_{\ell-1} \cdots x_{1} \in \{0, 1\}^\ell$ be a binary string.
We first introduce the following notations:
\begin{itemize}
    \item $\mathsf{Prefix}(x_{\ell} x_{\ell-1}\cdots x_{1},k) = x_{\ell} x_{\ell-1}\cdots x_{k+1}$
    \item $\mathsf{AllPrefix}(x_{\ell} x_{\ell-1}\cdots x_{1},k) = \{ x_{\ell} x_{\ell-1}\cdots x_{j+1} \}_{j\in [0,k]}$
    \item $\mathsf{UpBound}(x_{\ell} x_{\ell-1}\cdots x_{k})=x_{\ell} x_{\ell-1}\cdots x_{k} \| 11\cdots 1$
    \item $\mathsf{LowBound}(x_{\ell} x_{\ell-1}\cdots x_{k})=x_{\ell} x_{\ell-1}\cdots x_{k} \| 00\cdots 0$
    \item $\mathsf{Interval}(x_{\ell} x_{\ell-1}\cdots x_{k}) = \{x_{\ell}\cdots x_{k} \| x^*\}_{x^* \in \{0,1\}^{k-1}}$
\end{itemize}

Previous works~\cite{dang2025ccs,van2025} gave the definition of $\mathsf{Decompose}$ as following:
\begin{definition}
Given an integer interval $[a-\delta,a+\delta]$, there is an algorithm $\mathsf{Decompose}$ to succinctly encode the interval into a list of prefixes $\{p_i\}_{i\in [\hat{\mu}]}$ such that:
\begin{itemize}
    \item $[a-\delta,a+\delta] = \bigcup_{i\in [\hat{\mu}]} 
    \mathsf{Interval}(p_i)$
    \item For any $i\in [\hat{\mu}]$, if a binary string $p$ is a prefix of $p_i$, $\mathsf{Interval}(p_i) \nsubseteq [a-\delta,a+\delta]$
\end{itemize}
\end{definition}
The algorithm $\mathsf{Decompose}$~\cite{dang2025ccs,van2025} has a computation complexity $\bigo{\log\delta}$ and the number of prefixes $\hat{\mu}$ is also $\bigo{\log\delta}$. The length of prefixes $p_i$ is at least $\ell -\mu^\prime$ where $\mu^\prime$ is the number of wildcards (i.e., non-determined bit strings) and $\mu^\prime = \bigo{\log\delta}$. Precisely, if $\delta$ is a power of 2, the number of prefixes $\hat{\mu}$ has a lower bound of $\ceil{\log(2\delta+1)}$ and $\mu^\prime$ has an upper bound of $\floor{\log(2\delta+1)}$~\cite{van2025}. In this paper, we set our threshold $\delta$ to a power of 2 for simplicity. Unless otherwise specified, we denote $\hat{\mu} = 2+ \log\delta$ and $\mu^\prime = 1+\log\delta$ or $\log\delta$. For a binary string $x\in\{0,1\}^{\ell}$, there are $\mu^\prime+1$ prefixes of length at least $\ell-\mu^\prime$ and we denote $\mu =\mu^\prime+1$.

\section{OPPRF with Shared Outputs}
\label{sec: OPPRF with Shared Output}

We introduce a new variant of OPPRF, termed so-OPPRF, whose outputs are secret-shared between two parties. So-OPPRF is useful in scenarios where outputs cannot be revealed to either party individually, which may be of independent interest.
As we have presented our idea in Section~\ref{subsec: OPPRF with shared outputs},
here we illustrate the ideal functionality in Figure~\ref{Func:so-opprf}
and the detailed protocol in Figure~\ref{Prot:so-opprf}.

We show the protocol's security in Theorem~\ref{thm: Proof of so-OPPRF}, and the security proof is deferred to Appendix~\ref{subsec: Proof of so-OPPRF}.

\begin{theorem}
\label{thm: Proof of so-OPPRF}
    The protocol \Prot[so\text{-}OPPRF] in Figure \ref{Prot:so-opprf} realizes the functionality \Func[so\text{-}OPPRF] in Figure \ref{Func:so-opprf} against semi-honest adversaries in the (\Func[so\text{-}OPRF])-hybrid model if $\mathsf{OKVS}$ satisfies the correctness, double obliviousness, and independence properties defined in Definition \ref{def: okvs}.
\end{theorem}

\section{Efficient Fuzzy Mapping}
\label{sec: Fuzzy Mapping}

\begin{figure}[t]
\begin{nfproc}{$\mathsf{LocalMap}$}

\noindent \textbf{Parameters:} Distance threshold is \( \delta \). Prefix parameter $\hat{\mu}=2+\log\delta$.

\noindent \textbf{Input:} A set \( Q=\{q_j\}_{j\in [m]} \in \UU^{d\times m} \).

\noindent \textbf{Protocol}:

 \begin{enumerate}
    \item Initialize $\mathsf{List} := \emptyset$ and $\mathsf{interval}_k := \emptyset$ for $k \in [d]$.
        \item For each $j \in [m]$ and each $k \in [d]$:
        \begin{itemize}
            \item  Sample $r_{j,k}$ and define $U_{j,k} := [q_{j,k}-\delta, q_{j,k}+\delta]$.

            \item For each $(U,r) \in \mathsf{interval}_k$ if $U_{j,k} \cap U \neq \emptyset$,  remove $(U,r)$ from $\mathsf{interval}_k$ and update $U_{j,k}=U_{j,k}\cup U$. 
            \item Set $\mathsf{interval}_k := \mathsf{interval}_k \cup \{(U_{j,k}, r_{j,k})\}$.
        \end{itemize}

    \item For each $j \in [m]$, set $\mathsf{pid}_{q_j} := \sum_{k=1}^d r_{k}$, where $r_{k}$ satisfies that $q_{j,k} \in U_{k}$ and $(U_{k}, r_{k}) \in \mathsf{interval}_k$.
    
    \item For each $k \in [d]$ and each $(U, r) \in \mathsf{interval}_k$:
        \begin{itemize}
            \item {Without prefix optimization}: For each $x \in U$, set  
            $\mathsf{List} := \mathsf{List} \cup \{(k\Vert x, r)\}$.
            \item {With prefix optimization}: Partition interval $U$ into consecutive, disjoint sub-intervals $\{U_j\}_{j \in [\epsilon]}$, such that $|U_j| = 2\delta+1$ for $j \in [\epsilon-1]$ and $|U_\epsilon| \le 2\delta+1$. For each $j \in [\epsilon]$, compute $\{x_{j,h}\}_{h \in [\hat{\mu}]} := \mathsf{Decompose}(U_j)$ and set  
            $\mathsf{List} := \mathsf{List} \cup \{(k\Vert x_{j,h}, 0\Vert r)\}_{h \in [\hat{\mu}]}$.

        \end{itemize}

    \item Pad $\mathsf{List}$ with dummy items so that its size equals
    \begin{itemize}
        \item {Without prefix optimization}: $m d (2\delta + 1)$.
        \item {With prefix optimization}: $m d (\log \delta + 2)$.
    \end{itemize}

    \item Output $\{ \mathsf{pid}_{q_j}\}_{j\in [m]}$ and $\mathsf{List}$.
\end{enumerate}

\end{nfproc}
\vspace{0.5em}
\caption{Procedure of local mapping and its variant with prefix optimization.}
\label{Proc:local-mapping-prefix}
\end{figure}

Fuzzy mapping (FMap), introduced by Gao et al.~\cite{gao2025efficient}, enables two parties to assign identifiers to their respective set elements. If an element from the sender \SSS and an element from the receiver \RRR are sufficiently close, they will be mapped to the same identifier. We recall the formal definition from Gao et al. \cite{gao2025efficient} as follows.

\begin{definition}[Fuzzy mapping~\cite{gao2025efficient}]
\label{def: fuzzy mapping}
    A two-party protocol \Prot[FMap], where \sender inputs \( Q=\{q_j\}_{j\in [m]} \in \UU^{d\times m} \) and learns $\{\mathsf{ID}_{q_j}\}_{j \in [m]} \in \FF^m$ and \receiver inputs \( W=\{w_i\}_{i\in [n]} \in \UU^{d\times n} \) and learns $\{\mathsf{ID}_{w_i}\}_{i \in [n]} \in \FF^n$, is a secure fuzzy mapping protocol for distance metric $\mathsf{dist}$ and threshold $\delta$ against semi-honest adversaries, if and only if it satisfies:

    \begin{enumerate}
        \item \textbf{Correctness:} For any \( w \in W \) and \( q \in Q \), if \( \mathsf{dist}(w, q) \leq \delta \), then \( \mathsf{ID}_{q} = \mathsf{ID}_{w} \).

        \item \textbf{Distinctiveness:} For any \( w_i,w_{i^\prime} \in W\) and $i\neq {i^\prime}$, $\mathsf{Pr}[\mathsf{ID}_{w_i} = \mathsf{ID}_{w_{i^\prime}}] \leq \mathsf{negl}(\lambda)$. 

        \item \textbf{Security:} Considering a corrupted sender \SSS, for any $Q \in \UU^{d \times m}$ and any $W, W^\prime \in \UU^{d \times n}$, it holds that $ \mathsf{view}^\Pi_{\SSS}(Q, W) \approx_c \mathsf{view}^\Pi_{\SSS}(Q, W^\prime)$. Similarly, considering a corrupted receiver \RRR, for any $W \in \UU^{d \times n}$ and any $Q, Q^\prime \in \UU^{d \times m}$, it holds that $\mathsf{view}^\Pi_{\RRR}(Q, W) \approx_c \mathsf{view}^\Pi_{\RRR}(Q^\prime, W)$.
    
    \end{enumerate}
\end{definition}

As we have presented our idea in Section~\ref{subsec: Fuzzy mapping from so-OPPRF and si-OPRF}, here we give the construction of the fuzzy mapping protocol for $L_\infty$ distance in Figure~\ref{Prot:fuzzy-map} and its sub-procedure local mapping in Figure~\ref{Proc:local-mapping-prefix}.
We note that for any $w, q$, the fact that $\mathsf{dist}_\infty(q, w) \le \mathsf{dist}_p(q, w)$ holds. Therefore, we will use fuzzy mapping for $L_\infty$ distance in fuzzy PSI protocols for both $L_\infty$ and $L_p$ distance, since fuzzy mapping can tolerate false positives.

\begin{figure}[t]
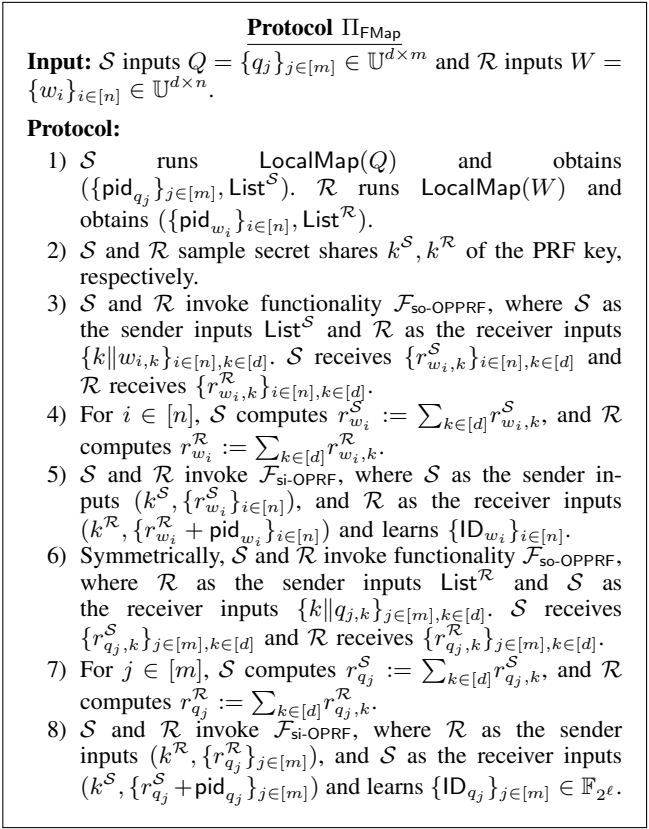

    \begin{nfprot}{\Prot[FMap]}

        \noindent \textbf{Input:} \sender inputs \( Q=\{q_j\}_{j\in [m]} \in \UU^{d\times m} \) and \receiver inputs \( W=\{w_i\}_{i\in [n]} \in \UU^{d\times n} \).

        \vspace{0.5em}
        \noindent \textbf{Protocol:}

        \begin{enumerate}

            \item \SSS runs $\mathsf{LocalMap}(Q)$ and obtains $(\{\mathsf{pid}_{q_j}\}_{j \in [m]}, \mathsf{List}^\SSS)$. \RRR runs $\mathsf{LocalMap}(W)$ and obtains $(\{\mathsf{pid}_{w_i}\}_{i \in [n]}, \mathsf{List}^\RRR)$.

            \item \SSS and \RRR sample secret shares $k^\SSS, k^\RRR$ of the PRF key, respectively.

             \item \sender and \receiver invoke functionality \Func[so\text{-}OPPRF], where \SSS as the sender inputs $\mathsf{List}^\SSS$ and \RRR as the receiver inputs $\{k\Vert w_{i,k}\}_{i\in [n],k\in [d]}$. \SSS receives $\{ r^{\SSS}_{w_i,k} \}_{i\in[n],k\in [d]}$ and \RRR receives $\{ r^{\RRR}_{w_i,k} \}_{i\in[n],k\in [d]}$.

            \item For \( i \in [n] \), \SSS computes $r^{\SSS}_{w_i} := {\sum}_{k\in [d]} r^{\SSS}_{w_i,k}$, and \RRR computes $r^{\RRR}_{w_i} := {\sum}_{k\in [d]} r^{\RRR}_{w_i,k}$.

            \item \SSS and \RRR invoke \Func[si\text{-}OPRF], where \SSS as the sender inputs $(k^\SSS, \{r^{\SSS}_{w_i}\}_{i \in [n]})$, and \RRR as the receiver inputs $(k^\RRR, \{r^{\RRR}_{w_i} + \mathsf{pid}_{w_i}\}_{i \in [n]})$ and learns $\{\mathsf{ID}_{w_i}\}_{i \in [n]}$.

            \item Symmetrically, \sender and \receiver invoke functionality \Func[so\text{-}OPPRF], where \RRR as the sender inputs $\mathsf{List}^\RRR$ and \SSS as the receiver inputs $\{k\Vert q_{j,k}\}_{j\in [m],k\in [d]}$. \SSS receives $\{ r^{\SSS}_{q_j,k} \}_{j\in [m],k\in [d]}$ and \RRR receives $\{ r^{\RRR}_{q_j,k} \}_{j\in [m],k\in [d]}$.

            \item For \( j \in [m] \), \SSS computes $r^{\SSS}_{q_j} := {\sum}_{k\in [d]} r^{\SSS}_{q_j,k}$, and \RRR computes $r^{\RRR}_{q_j} := {\sum}_{k\in [d]} r^{\RRR}_{q_j,k}$.

            \item \SSS and \RRR invoke \Func[si\text{-}OPRF], where \RRR as the sender inputs $(k^\RRR, \{r^{\RRR}_{q_j}\}_{j \in [m]})$, and \SSS as the receiver inputs $(k^\SSS, \{r^{\SSS}_{q_j} + \mathsf{pid}_{q_j}\}_{j \in [m]})$ and learns $\{\mathsf{ID}_{q_j}\}_{j \in [m]} \in \FF_{2^\ell}$.

        \end{enumerate}
    \end{nfprot}
    \vspace{0.5em}
    \caption{Protocol of fuzzy mapping for $L_\infty$ distance.}
    \label{Prot:fuzzy-map}
\end{figure}

\begin{figure}[t]
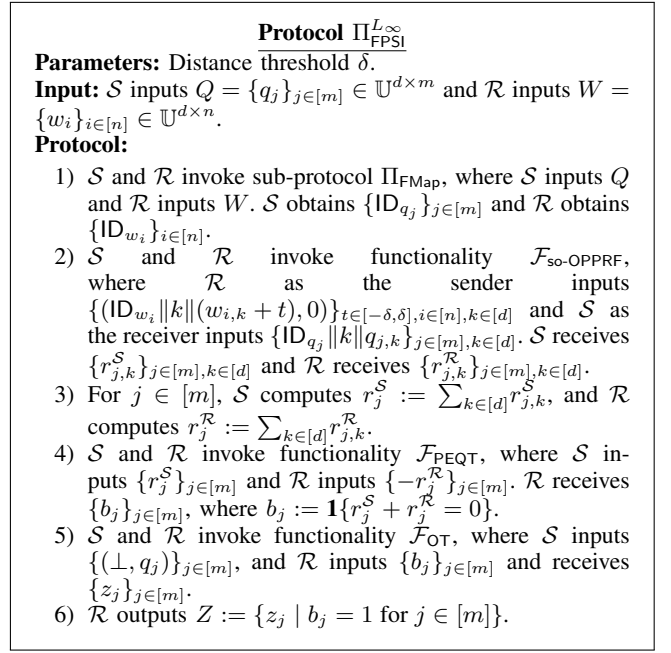

    \begin{nfprot}{\ensuremath{\Pi_\mathsf{FPSI}^{L_\infty}}\xspace}
        \noindent \textbf{Parameters:}  Distance threshold \( \delta \). 

        \noindent \textbf{Input:} \SSS inputs \( Q=\{q_j\}_{j\in [m]} \in \mathbb{U}^{d\times m} \) and \RRR inputs \( W=\{w_i\}_{i\in [n]} \in \mathbb{U}^{d\times n} \).

        \noindent \textbf{Protocol:}

        \begin{enumerate}
            \item \sender and \receiver invoke sub-protocol \Prot[FMap], where \sender inputs $Q$ and \receiver inputs $W$. \sender obtains  $\{\mathsf{ID}_{q_j} \}_{j\in [m]}$ and \receiver obtains $\{\mathsf{ID}_{w_i} \}_{i\in [n]}$.

            \item \sender and \receiver invoke functionality \Func[so\text{-}OPPRF], where \RRR as the sender inputs $\{\left( \mathsf{ID}_{w_i} \Vert k\Vert (w_{i,k}+t), 0 \right)\}_{t\in [-\delta,\delta], i \in [n],k\in [d]}$ and \SSS as the receiver inputs $\{ \mathsf{ID}_{q_j} \Vert k \Vert q_{j,k} \}_{j \in [m],k \in [d]}$. \SSS receives $\{ r^{\SSS}_{j,k} \}_{j\in [m],k\in [d]}$ and \RRR receives $\{ r^{\RRR}_{j,k} \}_{j\in[m],k\in [d]}$.

            \item For \( j \in [m] \), \SSS computes $r^{\SSS}_{j} :={\sum}_{k\in [d]} r^{\SSS}_{j,k}$, and \RRR computes $r^{\RRR}_{j} := {\sum}_{k\in [d]} r^{\RRR}_{j,k}$.

            \item \SSS and \RRR invoke functionality \Func[PEQT], where \SSS inputs $\{r^{\SSS}_{j}\}_{j \in [m]}$ and \RRR inputs $\{-r^{\RRR}_{j}\}_{j \in [m]}$. \RRR receives $\{b_j\}_{j \in [m]}$, where $b_j := \textbf{1}\{r^{\SSS}_{j}+r^{\RRR}_{j}=0\}$.
            
            \item \SSS and \RRR invoke functionality \Func[OT], where \SSS inputs $\{(\bot ,q_j)\}_{j \in [m]}$, and \RRR inputs $\{b_j\}_{j \in [m]}$ and receives $\{z_j\}_{j \in [m]}$.
            
            \item \RRR outputs $Z := \{z_j \mid b_j = 1\;\text{for}\; j\in [m]\}$.
        \end{enumerate}
    \end{nfprot}
    \vspace{0.5em}
    \caption{Protocol of fuzzy PSI for $L_\infty$ distance.}
    \label{Prot:fpsi1}
\end{figure}

We show that the above protocol satisfies the fuzzy mapping definition in Theorem \ref{thm: fuzzy mapping}.

\begin{theorem}
\label{thm: fuzzy mapping}
    The \Prot[FMap] protocol in Figure~\ref{Prot:fuzzy-map} satisfies the correctness, distinctiveness, and security properties defined in Definition \ref{def: fuzzy mapping} for $L_\infty$ distance, if both parties' sets satisfy the disjoint projection assumption.
\end{theorem}

\begin{proof}
\label{proof: fuzzy mapping}
\textbf{Correctness.}  
For any \( w_i \in W \) and \( q_j \in Q \), if \( \mathsf{dist}(w_i, q_j) \leq \delta \), then for any $k \in [d]$, $\abs{w_{i,k} - q_{j,k}} \leq \delta$ always holds. Therefore, all these $w_{i,k}$ are assigned in $\mathsf{List}^\SSS$ and $r^{\SSS}_{w_{i, k}} + r^{\RRR}_{w_{i, k}} = \mathsf{List}^\SSS[k\|w_{i, k}]$ according to the functionality \Func[so\text{-}OPPRF].
Moreover, the $\mathsf{LocalMap}$ procedure ensures that the $2\delta+1$ points on dimension $k$ centered at $q_{j,k}$ all have the same assignment, hence $\sum_{k \in [d]} (r^{\SSS}_{w_{i, k}} + r^{\RRR}_{w_{i, k}}) = \mathsf{pid}_{q_j}$ holds.
Similarly, it holds that $\sum_{k \in [d]} (r^{\SSS}_{q_{j, k}} + r^{\RRR}_{q_{j, k}}) = \mathsf{pid}_{w_i}$.
The identifiers are computed as follows $\mathsf{ID}_{q_j} := \mathsf{PRF}_k(\sum_{k \in [d]} (r^{\SSS}_{q_{j, k}} + r^{\RRR}_{q_{j, k}}) + \mathsf{pid}_{q_j})$ and $\mathsf{ID}_{w_i} := \mathsf{PRF}_k(\sum_{k \in [d]} (r^{\SSS}_{w_{i, k}} + r^{\RRR}_{w_{i, k}}) + \mathsf{pid}_{w_i})$.
Since the PRF key is the same for both PRF evaluations on $q_j$ and $w_i$, then \( \mathsf{ID}_{q_j} = \mathsf{ID}_{w_i} \) when \( \mathsf{dist}(w_i, q_j) \leq \delta \).

\textbf{Distinctiveness.} Since the \RRR's set $W$ satisfies the disjoint projection assumption, without loss of generality, for each $w_i$, assume the interval of $[w_{i, k_i}-\delta, w_{i, k_i}+\delta]$ on dimension $k_i$ is disjoint from the length-($2\delta+1$) intervals of other elements of $W$. Thus, the assignment $\mathsf{List}^\RRR[k_i\|w_{i, k_i}]$ is a uniformly random value and is independent of any other assignments. 
Then, we have that $\mathsf{ID}_{w_i} := \mathsf{PRF}(k, \mathsf{List}^\RRR[k_i\|w_{i, k_i}] + \ldots)$.
According to the functionality of si-OPRF, $\mathsf{ID}_{w_i}$ is computationally indistinguishable from a uniformly random value.
Therefore, a union bound shows that the probability of that there exists $i, i^\prime \in [n]$ such that $i \neq i^\prime$ and $\mathsf{ID}_{w_i} = \mathsf{ID}_{w_{i^\prime}} \in \FF$ is at most $n^2/|\FF|$.
With $|\FF| \geq n^2 \cdot 2^\lambda$, the probability $n^2/|\FF| \leq 1/2^\lambda$ is negligible. 

\textbf{Security.}
Since the protocol is symmetric, we only consider the corrupted sender \SSS.
Let $\mathsf{ID}_Q(W)$ denote the \SSS's output from the real-world protocol where \SSS inputs $Q$ and \RRR inputs $W$.
We show the simulator $\Sim[\SSS]^{\mathsf{FMap}}(Q, \mathsf{ID}_Q(W))$:
    \begin{enumerate}
        \item $\Sim[\SSS]^{\mathsf{FMap}}$ runs $\mathsf{LocalMap}(Q)$ and records the sampled randomness. 

        \item $\Sim[\SSS]^{\mathsf{FMap}}$ randomly samples the PRF key share $k^\SSS$.

        \item $\Sim[\SSS]^{\mathsf{FMap}}$ randomly samples $\{ r^{\SSS}_{w_i,k} \}_{i\in[n],k\in [d]}$.

        \item $\Sim[\SSS]^{\mathsf{FMap}}$ randomly samples $\{ r^{\SSS}_{q_j,k} \}_{j\in [m],k\in [d]}$.

        \item $\Sim[\SSS]^{\mathsf{FMap}}$ appends the above values and $(Q, \mathsf{ID}_Q(W))$ into the simulated view.
    \end{enumerate}
    
We show that the output by $\Sim[\SSS]^{\mathsf{FMap}}$ is indistinguishable from the real protocol. The only difference is $\{ r^{\SSS}_{w_i,k} \}_{i\in[n],k\in [d]}$ and $\{ r^{\SSS}_{q_j,k} \}_{j\in [m],k\in [d]}$. They are random secret shares in the real protocol according to the so-OPPRF functionality, while they are randomly sampled in the simulated view with the same distribution.
Therefore, for any $Q, W$, it holds that $ \mathsf{view}^\Pi_{\SSS}(Q, W) \approx_c \Sim[\SSS]^{\mathsf{FMap}}(Q, \mathsf{ID}_Q(W))$.

Moreover, since the \SSS's set $Q$ satisfies the disjoint projection assumption, as shown in the above analysis of distinctiveness, $\mathsf{ID}_{Q}$ is computationally indistinguishable from uniformly random values for any $W$.
Then, for any $Q, W$, it holds that $\Sim[\SSS]^{\mathsf{FMap}}(Q, \mathsf{ID}_Q(W)) \approx_c \Sim[\SSS]^{\mathsf{FMap}}(Q, R \from \FF^m)$.
Therefore, we have $\mathsf{view}^\Pi_{\SSS}(Q, W) \approx_c \Sim[\SSS]^{\mathsf{FMap}}(Q, R \from \FF^m) \approx_c \mathsf{view}^\Pi_{\SSS}(Q, W^\prime)$.
This completes the proof.

\end{proof}

\section{Fuzzy PSI}
\label{sec: Fuzzy-PSI}

\subsection{Fuzzy PSI for $L_\infty$ Distance}
\label{sec: Fuzzy-PSI l infty}

We have presented our idea in Section~\ref{subsec: Fuzzy PSI from so-OPPRF} and the detailed fuzzy PSI protocol for $L_\infty$ distance is illustrated in Figure~\ref{Prot:fpsi1}. 
We show the protocol's correctness as follows.

\begin{proof}
\noindent \textbf{Correctness.}
For each $q_j \in Q$, if there exists $w_i \in W$ such that $\mathsf{dist}(q_j, w_i) \leq \delta$, the correctness of the fuzzy mapping in Theorem~\ref{thm: fuzzy mapping} guarantees that $\mathsf{ID}_{q_j} = \mathsf{ID}_{w_i}$. Consequently, by the correctness of the functionalities \Func[so\text{-}OPPRF], \Func[PEQT], and \Func[OT], the receiver correctly obtains $b_j = 1$ and the corresponding element $q_j$.

On the other hand, for each $q_j \in Q$, if $\mathsf{dist}(q_j, w_i) > \delta$ for all $w_i \in W$, two cases arise due to the false positives introduced by the fuzzy mapping.
(1) The first case is $\mathsf{ID}_{q_j} \neq \mathsf{ID}_{w_i}$ for any $w_i \in W$. According to the correctness of functionality \Func[so\text{-}OPPRF], the secret-shared outputs $r_{j,k} \in \FF$ for $k\in [d]$ are uniformly random.
(2) The second case is $\mathsf{ID}_{q_j} = \mathsf{ID}_{w_i}$ for some $w_i \in W$ and there exists some $k^*$ such that $\abs{q_{j, k^*} - w_{i, k^*}} > \delta$. According to the correctness of functionality \Func[so\text{-}OPPRF], the secret-shared output $r_{j,k^*} \in \FF$ is uniformly random.
In both cases, $r_{j} :={\sum}_{k\in [d]} r_{j,k} \in \FF$ is uniformly random.
According to the correctness of functionality \Func[PEQT], a union bound shows that the probability of there existing $j \in [m]$ such that $b_j = 1$ is at most $m/|\FF|$.
By setting $|\FF| \geq m \cdot 2^\lambda$, the probability $m/|\FF| \leq 1/2^\lambda$ is negligible. 
The correctness of \Func[OT] ensures the receiver learns nothing about $q_j$.
\end{proof}

We show the protocol's security in Theorem \ref{thm: Proof of fuzzy PSI L-inf} and the security proof is deferred to Appendix \ref{subsec: Proof of fuzzy PSI L-inf}.

\begin{theorem}
\label{thm: Proof of fuzzy PSI L-inf}
    The protocol \ensuremath{\Pi_\mathsf{FPSI}^{L_\infty}} in Figure \ref{Prot:fpsi1} realizes the functionality \Func[FPSI] for $L_\infty$ distance in Figure \ref{Func:FPSI} against semi-honest adversaries in the $(\Func[so\text{-}OPPRF], \Func[PEQT], \Func[OT])$-hybrid model.
\end{theorem}

\subsection{Fuzzy PSI for $L_p$ Distance}

\begin{figure}[t]
    \begin{nfprot}{$\Pi_\mathsf{FPSI}^{L_p}$}
        \noindent \textbf{Parameters:} Distance threshold \( \delta \).

        \noindent \textbf{Input:} \SSS inputs \( Q=\{q_j\}_{j\in [m]} \in \mathbb{U}^{d\times m} \) and \RRR inputs \( W=\{w_i\}_{i\in [n]} \in \mathbb{U}^{d\times n} \).

        \noindent \textbf{Protocol:}

        \begin{enumerate}
             \item \sender and \receiver invoke sub-protocol \Prot[FMap], where \sender inputs $Q$ and \receiver inputs $W$. \sender obtains  $\{\mathsf{ID}_{q_j} \}_{j \in [m]}$ and \receiver obtains $\{\mathsf{ID}_{w_i} \}_{i\in [n]}$.

            \item \sender and \receiver invoke functionality \Func[so\text{-}OPPRF], where \RRR as the sender inputs $\{\left( \mathsf{ID}_{w_i} \Vert k\Vert (w_{i,k}+t), \abs{t}^p \right)\}_{t\in [-\delta,\delta], i \in [n],k\in [d]}$ and \SSS as the receiver inputs $\{ \mathsf{ID}_{q_j} \Vert k \Vert q_{j,k} \}_{j \in [m],k \in [d]}$. \SSS receives $\{ r^{\SSS}_{j,k} \}_{j\in [m],k\in [d]}$ and \RRR receives $\{ r^{\RRR}_{j,k} \}_{j\in[m],k\in [d]}$.

            \item (Optional) \SSS and \RRR invoke functionality \Func[B2A], where \SSS inputs $\{r^{\SSS}_{j,k}\}_{j \in [m], k \in [d]}$ and \RRR inputs $\{r^{\RRR}_{j,k}\}_{j \in [m], k \in [d]}$. \SSS receives $\{d^{\SSS}_{j,k}\}_{j \in [m], k \in [d]}$ and \RRR receives $\{d^{\RRR}_{j,k}\}_{j \in [m], k \in [d]}$.

            \item For \( j \in [m] \), \SSS computes $d^{\SSS}_{j} := {\sum}_{k\in [d]} d^{\SSS}_{j,k}$, and \RRR computes $d^{\RRR}_{j} := {\sum}_{k\in [d]} d^{\RRR}_{j,k}$.

            \item \SSS and \RRR invoke functionality \Func[Interval], where \SSS inputs $\{d^{\SSS}_{j}\}_{j \in [m]}$ and \RRR inputs $\{d^{\RRR}_{j}\}_{j \in [m]}$. \RRR receives $\{b_{j}\}_{j \in [m]}$, where $b_j := \textbf{1}\{d^{\SSS}_{j}+d^{\RRR}_{j} \le \delta^p\}$.
            
            \item \SSS and \RRR invoke functionality \Func[OT], where \SSS inputs $\{(\bot ,q_j)\}_{j \in [m]}$, and \RRR inputs $\{b_j\}_{j \in [m]}$ and receives $\{z_j\}_{j \in [m]}$.
            
            \item \RRR outputs $Z := \{z_j \mid b_j = 1\;\text{for}\; j\in [m]\}$.
        \end{enumerate}
    \end{nfprot}
    \vspace{0.5em}
    \caption{Protocol of fuzzy PSI for $L_p$ distance.}
    \label{Prot:fpsi2}
\end{figure}

We present the construction of fuzzy PSI for $L_p$ distance, which closely resembles that of $L_\infty$ distance.
Given the fact that $\mathsf{dist}_p(q, w) \le \delta$ implies $\mathsf{dist}_\infty(q, w) \le \delta$, this protocol reuses the fuzzy mapping protocol for $L_\infty$ distance to generate IDs, with the only modification in the refined filtering.

Specifically, two parties invoke functionality \Func[so\text{-}OPPRF], where the receiver sets the key-value pairs as $\{\left(\mathsf{ID}_{w_{i}} \Vert k \Vert (w_{i, k} + t),\abs{t}^p\right)\}_{i \in [n], k \in [d], t \in [-\delta, \delta]}$, and the sender prepares $\{ \mathsf{ID}_{q_j} \Vert k \Vert q_{j,k} \}_{j \in [m],k \in [d]}$.
Then, two parties get the secret sharing of the distance $| w_{i, k} - q_{j, k}|^p$ for each dimension $k \in [d]$ or a random value. 
The following process computes the distance $\mathsf{dist}_p(q_j, w_i)^p = \sum_{k \in [d]} \abs{w_{i, k} - q_{j, k}}^p$.
We note that the protocol includes a Boolean-to-arithmetic conversion (B2A) protocol, since the outputs of our so-OPPRF protocol reside in a binary field $\mathbb{F}_{2^\ell}$ that does not support computing arithmetic operations of the $L_p$ distance.
B2A converts the outputs of so-OPPRF into arithmetic shares, which are suitable for subsequent computation.

In addition, unlike the private equality test used in protocol  \ensuremath{\Pi_\mathsf{FPSI}^{L_\infty}}, here we invoke the private interval test protocol~\cite{chakraborti2023distance} to determine whether the distance falls below the threshold $\delta^p$. The detailed construction of the protocol is illustrated in Figure~\ref{Prot:fpsi2}. We show the correctness as follows.

\begin{figure}[t]
    \begin{nfprot}{\Prot[EQSel]}
    \noindent \textbf{Parameters:} Ideal functionality \Func[MUX] and \Func[ssPEQT].
    
    \noindent \textbf{Input:} \sender inputs $\{(e_i^\SSS, v_i^\SSS)\}_{i \in [h]}$ and \receiver inputs $\{(e_i^\RRR, v_i^\RRR)\}_{i \in [h]}$.
    
    \noindent \textbf{Protocol:}
    \begin{enumerate}
        \item \SSS and \RRR invoke \Func[ssPEQT], where \SSS inputs $\{e_i^\SSS\}_{i \in [h]}$ and \RRR inputs $\{e_i^\RRR\}_{i \in [h]}$. \SSS receives $\{b_i^\SSS\}_{i \in [h]} \in \bool^h$, and \RRR receives $\{b_i^\RRR\}_{i \in [h]} \in \bool^h$. 

        \item \SSS and \RRR invoke \Func[MUX], where \SSS inputs $\{(b_i^\SSS, v_i^\SSS)\}_{i \in [h]}$ and \RRR inputs $\{(b_i^\RRR, v_i^\RRR)\}_{i \in [h]}$. \SSS receives $\{t_i^\SSS\}_{i \in [h]}$ and \RRR receives $\{t_i^\RRR\}_{i \in [h]}$.

        \item \SSS computes $b^\SSS := \Xor_{i \in [h]} b_i^\SSS$, $t^\SSS := \sum_{i \in [h]} t_i^\SSS$ and \RRR computes $b^\RRR := \Xor_{i \in [h]} b_i^\RRR$, $t^\RRR := \sum_{i \in [h]} t_i^\RRR$.
        
        \item \SSS randomly samples $r^\SSS$ and \RRR randomly samples $r^\RRR$.
        
        \item \SSS and \RRR invoke \Func[MUX], where \SSS inputs $(b^\SSS, t^\SSS - r^\SSS)$ and \RRR inputs $(b^\RRR, t^\RRR - r^\RRR)$. \SSS receives $m^\SSS$ and \RRR receives $m^\RRR$.

        \item \SSS outputs $z^\SSS := m^\SSS + r^\SSS$. \RRR outputs $z^\RRR := m^\RRR + r^\RRR$.

    \end{enumerate}
    
    \end{nfprot}
    \vspace{0.5em}
    \caption{Protocol of random equality-conditional selection.}
    \label{Prot:select}
\end{figure}

\begin{figure}[t]
    \begin{nfprot}{\Prot[FMap\text{-}Prefix]}
        \noindent \textbf{Parameters:} Threshold $\delta$. Prefix parameter $\mu^\prime = 1+\log\delta$ and $\mu=2+\log\delta$.

        \noindent \textbf{Input:} \sender inputs \( Q=\{q_j\}_{j\in [m]} \in \UU^{d\times m} \) and \receiver inputs \( W=\{w_i\}_{i\in [n]} \in \UU^{d\times n} \).

        \vspace{0.5em}
        \noindent \textbf{Protocol:}

        \begin{enumerate}

            \item \SSS runs $\mathsf{LocalMap^{Prefix}}(Q)$ and obtains $(\{\mathsf{pid}_{q_j}\}_{j \in [m]}, \mathsf{List}^\SSS)$.  \RRR runs $\mathsf{LocalMap^{Prefix}}(W)$ and obtains $(\{\mathsf{pid}_{w_i}\}_{i \in [n]}, \mathsf{List}^\RRR)$.

            \item \SSS and \RRR sample secret shares $k^\SSS, k^\RRR$ of the PRF key, respectively.

            \item For $j \in [m], k \in [d]$, \sender computes $\{q_{j,k, h}\}_{h\in [\mu]} := \mathsf{AllPrefix}(q_{j,k}, \mu^\prime)$.
            \sender and \receiver invoke functionality \Func[so\text{-}OPPRF], where \RRR as the sender inputs $\mathsf{List}^\RRR$ and \SSS as the receiver inputs $\{k\Vert q_{j,k,h}\}_{j\in [m],k\in [d],h\in [\mu]}$. \SSS receives $\{ e^{\SSS}_{q_j, k, h} \Vert v^{\SSS}_{q_j, k, h} \}_{j\in [m],k\in [d],h\in [\mu]}$ and \RRR receives $\{ e^{\RRR}_{q_j, k, h} \Vert v^{\RRR}_{q_j, k, h} \}_{j\in [m],k\in [d],h\in [\mu]}$.

            \item For $j\in [m],k\in [d]$, \SSS and \RRR invoke sub-protocol $\Prot[EQSel]$, where \SSS inputs $\{(e^{\SSS}_{q_j, k, h}, v^{\SSS}_{q_j, k, h})\}_{h\in [\mu]}$ and receives $r^{\SSS}_{q_j, k}$, and \RRR inputs $\{(e^{\RRR}_{q_j, k, h},v^{\RRR}_{q_j, k, h})\}_{h\in [\mu]}$ and receives $r^{\RRR}_{q_j, k}$. 

            \item For \( j \in [m] \), \SSS computes $r^{\SSS}_{q_j} := {\sum}_{k\in [d]} r^{\SSS}_{q_j, k}$, and \RRR computes $r^{\RRR}_{q_j} := {\sum}_{k\in [d]} r^{\RRR}_{q_j, k}$.

            \item \SSS and \RRR invoke \Func[si\text{-}OPRF], where \RRR as the sender inputs $(k^\RRR, \{r^{\RRR}_{q_j}\}_{j \in [m]})$, and \SSS as the receiver inputs $(k^\SSS, \{r^{\SSS}_{q_j} + \mathsf{pid}_{q_j}\}_{j \in [m]})$ and learns $\{\mathsf{ID}_{q_j}\}_{j \in [m]}$.

            \item Symmetrically, for $i \in [n], k \in [d]$, \RRR computes $\{w_{i,k, h}\}_{h\in [\mu]} := \mathsf{AllPrefix}(w_{i,k}, \mu^\prime)$.
            \sender and \receiver invoke functionality \Func[so\text{-}OPPRF], where \SSS as the sender inputs $\mathsf{List}^\SSS$ and \RRR as the receiver inputs $\{k\Vert w_{i,k,h}\}_{i\in [n],k\in [d],h\in [\mu]}$. \SSS receives $\{ e^{\SSS}_{w_i, k, h} \Vert v^{\SSS}_{w_i, k, h} \}_{i\in [n],k\in [d],h\in [\mu]}$ and \RRR receives $\{ e^{\RRR}_{w_i, k, h} \Vert v^{\RRR}_{w_i, k, h} \}_{i\in [n],k\in [d],h\in [\mu]}$.

            \item For $i\in [n],k\in [d]$, \SSS and \RRR invoke sub-protocol $\Prot[EQSel]$, where \SSS inputs $\{(e^{\SSS}_{w_i, k, h}, v^{\SSS}_{w_i, k, h})\}_{h\in [\mu]}$ and receives $r^{\SSS}_{w_i, k}$, and \RRR inputs $\{(e^{\RRR}_{w_i, k, h},v^{\RRR}_{w_i, k, h})\}_{h\in [\mu]}$ and receives $r^{\RRR}_{w_i, k}$. 

            \item For \( i \in [n] \), \SSS computes $r^{\SSS}_{w_i} := {\sum}_{k\in [d]} r^{\SSS}_{w_i, k}$, and \RRR computes $r^{\RRR}_{w_i} := {\sum}_{k\in [d]} r^{\RRR}_{w_i, k}$.

            \item \SSS and \RRR invoke \Func[si\text{-}OPRF], where \SSS as the sender inputs $(k^\SSS, \{r^{\SSS}_{w_i}\}_{i \in [n]})$, and \RRR as the receiver inputs $(k^\RRR, \{r^{\RRR}_{w_i}  + \mathsf{pid}_{w_i}\}_{i \in [n]})$ and learns $\{\mathsf{ID}_{w_i}\}_{i \in [n]}$.

        \end{enumerate}
    \end{nfprot}
    \vspace{0.5em}
    \caption{Protocol of fuzzy mapping with prefix optimization.}
    \label{Prot:fuzzy map prefix}
\end{figure}

\begin{figure}[t]
    \begin{nfprot}{\ensuremath{\Pi_\mathsf{FPSI\text{-}Prefix}^{L_\infty}}\xspace}
        \noindent \textbf{Parameters:} Distance threshold \( \delta \). Prefix parameters $\mu^\prime = 1+\log\delta$ and $\mu = \hat{\mu} = 2+\log\delta$.

        \noindent \textbf{Input:} \SSS inputs \( Q=\{q_j\}_{j\in [m]} \in \mathbb{U}^{d\times m} \) and \RRR inputs \( W=\{w_i\}_{i\in [n]} \in \mathbb{U}^{d\times n} \).

        \vspace{0.5em}
        \noindent \textbf{Protocol:}

        \begin{enumerate}
            \item \sender and \receiver invoke sub-protocol \Prot[FMap\text{-}Prefix], where \sender inputs $Q$ and \receiver inputs $W$. \sender obtains  $\{\mathsf{ID}_{q_j} \}_{j\in [m]}$ and \receiver obtains $\{\mathsf{ID}_{w_i} \}_{i\in [n]}$.

            \item For $i\in [n],k\in [d]$, \RRR computes $\{w_{i,k,h}\}_{h\in [\hat{\mu}]} := \mathsf{Decompose}(w_{i,k}-\delta,w_{i,k}+\delta)$.
            \item For $j \in [m],k\in [d]$, \SSS computes $\{q_{j,k, h}\}_{h \in [\mu]} := \mathsf{AllPrefix}(q_{j,k},\mu^\prime)$.

            \item \sender and \receiver invoke functionality \Func[so\text{-}OPPRF], where \RRR inputs $\{\left( \mathsf{ID}_{w_i} \Vert k\Vert w_{i,k,h}, 0 \right)\}_{i \in [n],k\in [d], h \in [\hat{\mu}]}$ and \SSS inputs $\{ \mathsf{ID}_{q_j} \Vert k \Vert q_{j,k,h} \}_{j \in [m],k\in [d],h\in [\mu]}$. \SSS receives $\{ e^{\SSS}_{j,k,h} \Vert v^{\SSS}_{j,k,h} \}_{j\in [m],k\in [d],h\in [\mu]}$ and \RRR receives $\{ e^{\RRR}_{j,k,h} \Vert v^{\RRR}_{j,k,h} \}_{j\in[m],k\in [d],h\in [\mu]}$.

            \item For $j\in [m],k\in [d]$, \SSS and \RRR invoke sub-protocol $\Prot[EQSel]$, where \SSS inputs $\{(e^{\SSS}_{j, k, h}, v^{\SSS}_{j, k, h})\}_{h\in [\mu]}$ and receives $t^{\SSS}_{j, k}$, and \RRR inputs $\{(e^{\RRR}_{j, k, h},v^{\RRR}_{j, k, h})\}_{h\in [\mu]}$ and receives $t^{\RRR}_{j, k}$.

            \item For \( j \in [m] \), \SSS computes $t^{\SSS}_{j} := {\sum}_{k\in [d]} t^{\SSS}_{j,k}$, and \RRR computes $t^{\RRR}_{j} := {\sum}_{k\in [d]} t^{\RRR}_{j,k}$.

            \item \SSS and \RRR invoke functionality \Func[PEQT], where \SSS inputs $\{t^{\SSS}_{j}\}_{j \in [m]}$ and \RRR inputs $\{-t^{\RRR}_{j}\}_{j \in [m]}$. \RRR receives $\{b_j\}_{j \in [m]}$ where $b_j := \textbf{1}\{t^{\SSS}_{j}= - t^{\RRR}_{j}\}$.
            
            \item \SSS and \RRR invoke functionality \Func[OT], where \SSS inputs $\{(\bot ,q_j)\}_{j \in [m]}$, and \RRR inputs $\{b_j\}_{j \in [m]}$ and receives $\{z_j\}_{j \in [m]}$.
            
            \item \RRR outputs $Z := \{z_j \mid b_j = 1\;\text{for}\; j\in [m]\}$.
        \end{enumerate}
    \end{nfprot}
    \vspace{0.5em}
    \caption{Protocol of fuzzy PSI with prefix optimization for $L_\infty$ distance.}
    \label{Prot:fpsi1 prefix}
\end{figure}

\begin{proof}
\noindent\textbf{Correctness.} 
For each $q_j \in Q$, if there exists $w_i \in W$ such that $\mathsf{dist}_p(q_j, w_i) \leq \delta$, then by the correctness of fuzzy mapping for $L_\infty$ distance in Theorem \ref{thm: fuzzy mapping} and the fact that $\mathsf{dist}_\infty(q_j, w_i) \le \mathsf{dist}_p(q_j, w_i) \leq \delta$ for any $q_j, w_i$, it holds $\mathsf{ID}_{q_j} = \mathsf{ID}_{w_i}$. Then, by the correctness of functionalities \Func[so\text{-}OPPRF] and \Func[B2A], it holds $d^{\SSS}_{j, k} + d^{\RRR}_{j, k} = \abs{w_{i, k} - q_{j, k}}^p$, and hence $d^{\SSS}_{j} + d^{\RRR}_{j} = \sum_{k \in [d]} (d^{\SSS}_{j, k} + d^{\RRR}_{j, k}) = \sum_{k \in [d]} \abs{w_{i, k} - q_{j, k}}^p \leq \delta^p$.
Therefore, by the correctness of functionalities \Func[Interval] and \Func[OT], the receiver correctly obtains $b_j = 1$ and the corresponding element $q_j$.

On the other hand, for each $q_j \in Q$, if $\mathsf{dist}_p(q_j, w_i) > \delta$ in $L_p$ distance for all $w_i \in W$, there are three cases due to the false positive of fuzzy mapping.
(1) The first case is $\mathsf{ID}_{q_j} \neq \mathsf{ID}_{w_i}$ for any $w_i \in W$. According to the correctness of functionality \Func[so\text{-}OPPRF], the output $r^\SSS_{j,k} + r^\RRR_{j,k} \in \FF$ is uniformly random for $k\in [d]$.
Therefore, $r^S_{j} + r^R_{j} :={\sum}_{k\in [d]} (r^\SSS_{j,k} + r^\RRR_{j,k}) \in \FF$ is uniformly random.
(2) The second case is $\mathsf{ID}_{q_j} = \mathsf{ID}_{w_i}$ but $\mathsf{dist}_\infty(q_j,w_i) > \delta$. There exists $k^*\in[d]$ such that $|q_{j,k^*}-w_{i,k^*}|>\delta$. Then, according to the correctness of functionality \Func[so\text{-}OPPRF], the output $r^\SSS_{j,k^*} + r^\RRR_{j,k^*} \in \FF$ is uniformly random. Similarly, $r^S_{j} + r^R_{j} :={\sum}_{k\in [d]} (r^\SSS_{j,k} + r^\RRR_{j,k}) \in \FF$ is uniformly random.
As a result, in the first two cases, the probability of $r^S_{j} + r^R_{j}\le \delta$ is at most $\delta^p/|\FF|$. By the correctness of functionality \Func[interval], for any $j \in [m]$, a union bound shows that the probability of there existing $j \in [m]$ such that $b_j = 1$ is at most $m\delta^p/|\FF|$.
By setting $|\FF| \geq m \delta^p 2^\lambda$, the probability $m \delta^p/|\FF| \leq 1/2^\lambda$ is negligible. The correctness of functionality \Func[OT] ensures the receiver learns nothing about $q_j$. (3) The third case is $\mathsf{ID}_{q_j} = \mathsf{ID}_{w_i}$ and $\mathsf{dist}_\infty(q_j,w_i)\leq \delta $ for some $w_i \in W$. 
By the correctness of functionalities \Func[so\text{-}OPPRF] and \Func[B2A], $d^{\SSS}_{j, k} + d^{\RRR}_{j, k} = \abs{w_{i, k} - q_{j, k}}^p$, it holds $d^{\SSS}_{j} + d^{\RRR}_{j} = \sum_{k \in [d]} (d^{\SSS}_{j, k} + d^{\RRR}_{j, k}) = \sum_{k \in [d]} \abs{w_{i, k} - q_{j, k}}^p = \mathsf{dist}_p(q_j, w_i)^p > \delta^p$.
The correctness of \Func[Interval] and \Func[OT] ensure the receiver learns nothing about $q_j$.
\end{proof}

We show the protocol's security in Theorem~\ref{thm: Proof of fuzzy PSI L-p}, and the security proof is deferred to Appendix~\ref{subsec: Proof of fuzzy PSI L-p}.

\begin{theorem}
\label{thm: Proof of fuzzy PSI L-p}
    The protocol \ensuremath{\Pi_\mathsf{FPSI}^{L_p}} in Figure \ref{Prot:fpsi2} realizes the functionality \Func[FPSI] for $L_p$ distance in Figure \ref{Func:FPSI} against semi-honest adversaries in the $(\Func[so\text{-}OPPRF], \Func[B2A], \Func[Interval], \Func[OT])$-hybrid model.
\end{theorem}

\section{Fuzzy PSI with Prefix Optimization}
\label{sec: Fuzzy PSI with Prefix Optimization}

In this section, we present how to optimize our fuzzy PSI framework using prefix techniques \cite{garimella2024computation,van2025,dang2025ccs} for the large distance threshold $\delta$. In the previous section, our fuzzy PSI achieves linear complexity with respect to $\delta$ in both communication and computation. With prefix optimizations, we reduce the complexity to $\bigo{\log\delta}$.

\subsection{Optimized Fuzzy Mapping}
At a high level, by incorporating prefix techniques~\cite{garimella2024computation,van2025,dang2025ccs} into our framework, the sender in the so-OPPRF protocol does not need to program the entire interval $[q_{j,k}-\delta, q_{j,k}+\delta]$ of size $2\delta + 1$, but only $\bigo{\log\delta}$ prefixes that together cover the interval.
On the receiver side, this modification introduces an additional $\bigo{\log\delta}$ evaluations of so-OPPRF, since there are $\bigo{\log\delta}$ candidate prefixes that may match.
Consequently, the overall complexity is reduced from $O(\delta)$ to $O(\log \delta)$.

We first introduce a new building block, denoted as random equality-conditional selection (EQSel).
As illustrated in Figure~\ref{Prot:select}, EQSel takes as input a set of secret shares of $\{(e_i, v_i)\}_{i \in [h]}$ from two parties, where $e_i$ represents a condition and $v_i$ denotes a payload.
If there exists $i\in[h]$ such that $e_{i}=0$, the protocol outputs the associated payload $v_i$; otherwise, it outputs a random value.
Note that in our applications, there exists at most one $e_{i}$ equal to 0.
We instantiate this protocol using secret-shared private equality test~\cite{rindal2021vole} and secret-shared multiplexer~\cite{rathee2020cryptflow2}.

We present our optimized fuzzy mapping protocol in Figure~\ref{Prot:fuzzy map prefix}.
The main difference over our non-prefix fuzzy mapping protocol is that for each $w_{i, k}$, the receiver needs to prepare $\mu$ candidate prefixes for the invocation of so\text{-}OPPRF, where the sender sets the programmed values of $q_{j, k}$'s prefixes as $0\|r_{j, k}$.
Subsequently, both parties receive $\mu$ secret shares from so-OPPRF and then invoke EQSel on them to select $r_{j, k}$ or a random value depending on whether there is a matched prefix.
The remaining part of the protocol keeps the same as \Prot[FMap].

\begin{table*}[t]
\centering
\caption{The communication (MB) and running time (s) of fuzzy PSI for $L_\infty$, $L_1$ and $L_2$. “---” indicates out-of-memory. “Ours” denotes our fuzzy PSI protocols in Section~\ref {sec: Fuzzy-PSI}. The best results are marked in {\color{green!90}green}.}
\label{tab: low L0}
\resizebox{\textwidth}{!}{%
\begin{tabular}{cccccccccccccc}
\hline
\multirow{3}{*}{$m=n$} & \multirow{3}{*}{Protocol}          & \multicolumn{12}{c}{$(d,\delta)$}                                                                                                                                                         \\ \cline{3-14} 
                       &                                    & \multicolumn{2}{c}{$(4,16)$} & \multicolumn{2}{c}{$(8,16)$} & \multicolumn{2}{c}{$(16,16)$} & \multicolumn{2}{c}{$(4,32)$} & \multicolumn{2}{c}{$(8,32)$} & \multicolumn{2}{c}{$(16,32)$} \\ \cline{3-14} 
                       &                                    & Comm.         & Time         & Comm.         & Time         & Comm.          & Time         & Comm.         & Time         & Comm.         & Time         & Comm.          & Time         \\ \hline
                       
\multicolumn{14}{c}{$L_\infty$ \textbf{Distance}}                                                                                                                                                                                                                     \\ \hline
\multirow{3}{*}{$2^8$}   & Gao et al.~\cite{gao2025efficient}        & 23.2          & 4.4          & 46.2          & 8.6          & 92.1           & 16.6         & 45.5          & 7.8          & 90.7          & 15.1         & 181.2          & 30.5         \\
                       & Dang et al.~\cite{dang2025ccs} & 43.5          & 3.1          & 86.9          & 5.4          & 173.6          & 11.8         & 46.8          & 2.9          & 93.5          & 5.7          & 186.9          & 12.1         \\
                       & Ours                               & \cellcolor{green!25}7.3           & \cellcolor{green!25}0.3          & \cellcolor{green!25}10.6          & \cellcolor{green!25}0.3          & \cellcolor{green!25}17.1           & \cellcolor{green!25}0.4          & \cellcolor{green!25}9.4           & \cellcolor{green!25}0.3          & \cellcolor{green!25}14.7          & \cellcolor{green!25}0.4          & \cellcolor{green!25}25.2           & \cellcolor{green!25}0.5          \\ \hline
\multirow{3}{*}{$2^{12}$}  & Gao et al.~\cite{gao2025efficient}        & 370.9         & 73.2         & 738.5         & 144.6        & 1473.9         & 282.5        & 727.3         & 132.7        & 1451.3        & 253.9        & 2899.5         & 507.2        \\
                       & Dang et al.~\cite{dang2025ccs} & 695.8         & 43.3         & 1389.5        & 85.4         & 2777.0         & 181.2        & 749.3         & 46.4         & 1496.5        & 91.3         & 2990.9         & 190.7        \\
                       & Ours                               & \cellcolor{green!25}60.5          & \cellcolor{green!25}1.3          & \cellcolor{green!25}112.3         & \cellcolor{green!25}2.0          & \cellcolor{green!25}215.8         & \cellcolor{green!25}3.0          & \cellcolor{green!25}93.2          & \cellcolor{green!25}1.4          & \cellcolor{green!25}177.6         & \cellcolor{green!25}2.2          & \cellcolor{green!25}346.5         & \cellcolor{green!25}3.7          \\ \hline
\multirow{3}{*}{$2^{16}$} & Gao et al.~\cite{gao2025efficient}        & 5934.1        & 1195.7       & 11816.7       & 2269.5       & 23581.9        & 4637.7       & 11636.5       & 2132.4       & 23221.5       & 4197.1       & ---            & ---          \\
                       & Dang et al.~\cite{dang2025ccs} & 11132.6       & 707.3        & 22232.2       & 1406.1       & 44431.4        & 2945.4       & 11988.4       & 737.6        & 23943.8       & 1447.8       & 47854.6        & 3150.7       \\
                       & Ours                               &\cellcolor{green!25}908.7         &\cellcolor{green!25}17.2         &\cellcolor{green!25}1737.6        &\cellcolor{green!25}25.0         &\cellcolor{green!25}3396.0         &\cellcolor{green!25}42.0         &\cellcolor{green!25}1432.0        &\cellcolor{green!25}19.0         &\cellcolor{green!25}2784.7        &\cellcolor{green!25}28.9         &\cellcolor{green!25}5491.5         &\cellcolor{green!25}49.7           \\ \hline 
\multicolumn{14}{c}{$L_1$ \textbf{Distance}}                                                                                                                                                                                                                     \\ \hline 
\multirow{3}{*}{$2^8$}   & Gao et al.~\cite{gao2025efficient}        & 23.3          & 4.2          & 46.3          & 8.3          & 92.2           & 16.6         & 45.6          & 7.7          & 90.8          & 14.9         & 181.3          & 30.4         \\
                       & Dang et al.~\cite{dang2025ccs} & 61.2          & 3.8          & 122.2         & 7.2          & 244.2          & 14.7         & 71.2          & 4.0          & 142.2         & 7.6          & 284.2          & 15.6         \\
                       & Ours                               &\cellcolor{green!25}7.8           &\cellcolor{green!25}0.3          &\cellcolor{green!25}11.3          &\cellcolor{green!25}0.4          &\cellcolor{green!25}18.4           &\cellcolor{green!25}0.5          &\cellcolor{green!25}9.9           &\cellcolor{green!25}0.3          &\cellcolor{green!25}15.4          &\cellcolor{green!25}0.5          &\cellcolor{green!25}26.6           &\cellcolor{green!25}0.5          \\ \hline
\multirow{3}{*}{$2^{12}$}  & Gao et al.~\cite{gao2025efficient}        & 372.4         & 70.7         & 740.1         & 143.2        & 1475.4         & 282.6        & 729.2         & 129.6        & 1453.3        & 258.1        & 2901.4         & 506.9        \\ 
                       & Dang et al.~\cite{dang2025ccs} & 978.2         & 57.9         & 1954.8        & 116.7        & 3908.1         & 223.5        & 1138.3        & 61.3         & 2274.6        & 120.4        & 4547.1         & 244.3        \\ 
                       & Ours                               &\cellcolor{green!25}66.1          &\cellcolor{green!25}1.4          &\cellcolor{green!25}122.5         &\cellcolor{green!25}2.1          &\cellcolor{green!25}235.3          &\cellcolor{green!25}3.5          &\cellcolor{green!25}98.7          &\cellcolor{green!25}1.6          &\cellcolor{green!25}187.8         &\cellcolor{green!25}2.5          &\cellcolor{green!25}366.0          &\cellcolor{green!25}4.1          \\ \hline
\multirow{3}{*}{$2^{16}$} & Gao et al.~\cite{gao2025efficient}        & 5959.1        & 1210.6       & 11841.7       & 2299.7       & 23606.9        & 4593.3       & 11667.5       & 2091.3       & 23252.5       & 4165.0       & ---            & ---          \\
                       & Dang et al.~\cite{dang2025ccs} & 15650.7       & 930.8        & 31276.8       & 1833.3       & 62529.1        & 3646.7       & 18212.9       & 995.8        & 36393.2       & 1954.8       & 72753.9        & 3888.1       \\
                       & Ours                               &\cellcolor{green!25}990.9         &\cellcolor{green!25}17.9         &\cellcolor{green!25}1893.8        &\cellcolor{green!25}27.8         &\cellcolor{green!25}3700.2         &\cellcolor{green!25}46.9         &\cellcolor{green!25}1516.2        &\cellcolor{green!25}20.2         &\cellcolor{green!25}2942.9        &\cellcolor{green!25}31.5         &\cellcolor{green!25}5797.7         &\cellcolor{green!25}55.8         \\ \hline
\multicolumn{14}{c}{$L_2$ \textbf{Distance}}                                                                                                                                                                                                                     \\ \hline
\multirow{3}{*}{$2^8$}   & Gao et al.~\cite{gao2025efficient}        & 23.4          & 4.4          & 46.4          & 8.6          & 92.3           & 16.6         & 45.7          & 7.9          & 91.0          & 15.1         & 181.5          & 30.6         \\
                       & Dang et al.~\cite{dang2025ccs} & 70.1          & 4.0          & 140.0         & 7.6          & 279.6          & 14.5         & 84.8          & 4.6          & 169.0         & 8.9          & 337.4          & 16.0         \\
                       & Ours                               &\cellcolor{green!25}7.8           &\cellcolor{green!25}0.4          &\cellcolor{green!25}11.4          &\cellcolor{green!25}0.4          &\cellcolor{green!25}18.4           &\cellcolor{green!25}0.5          &\cellcolor{green!25}9.9           &\cellcolor{green!25}0.4          &\cellcolor{green!25}15.5          &\cellcolor{green!25}0.4          &\cellcolor{green!25}26.6           &\cellcolor{green!25}0.6          \\ \hline 
\multirow{3}{*}{$2^{12}$}  & Gao et al.~\cite{gao2025efficient}        & 373.9         & 72.7         & 741.6         & 143.2        & 1476.9         & 283.8        & 731.1         & 132.5        & 1455.2        & 259.9        & 2903.3         & 509.4        \\
                       & Dang et al.~\cite{dang2025ccs} & 1122.0        & 63.1         & 2239.4        & 117.7        & 4474.2         & 231.3        & 1356.0        & 72.6         & 2703.5        & 132.6        & 5398.4         & 253.3        \\
                       & Ours                               &\cellcolor{green!25}66.3          &\cellcolor{green!25}1.4          &\cellcolor{green!25}122.7         &\cellcolor{green!25}2.2          &\cellcolor{green!25}235.5          &\cellcolor{green!25}3.5          &\cellcolor{green!25}99.2          &\cellcolor{green!25}1.6          &\cellcolor{green!25}188.2         &\cellcolor{green!25}2.5          &\cellcolor{green!25}366.4          &\cellcolor{green!25}4.0          \\ \hline
\multirow{3}{*}{$2^{16}$} & Gao et al.~\cite{gao2025efficient}        & 5983.1        & 1223.9       & 11865.7       & 2321.1       & 23630.9        & 4585.5       & 11697.5       & 2127.9       & 23282.5       & 4134.8       & ---            & ---          \\
                       & Dang et al.~\cite{dang2025ccs} & 17951.5       & 1009.4       & 35830.4       & 1930.9       & 71588.3        & 3799.7       & 21696.1       & 1159.7       & 43255.6       & 2115.1       & 86374.7        & 4332.7       \\
                       & Ours                               &\cellcolor{green!25}998.8         &\cellcolor{green!25}18.6         &\cellcolor{green!25}1901.8        &\cellcolor{green!25}27.8         &\cellcolor{green!25}3708.2         &\cellcolor{green!25}47.7         &\cellcolor{green!25}1526.1        &\cellcolor{green!25}20.4         &\cellcolor{green!25}2952.8        &\cellcolor{green!25}31.7         &\cellcolor{green!25}5807.7         &\cellcolor{green!25}55.1         \\ \hline
\end{tabular}%
}
\end{table*}

\subsection{Optimized Fuzzy PSI for $L_{p\in[1,\infty]}$ Distance}

We present our fuzzy PSI protocol with prefix optimization for $L_\infty$ distance.
The optimized protocol for $L_p$ distance is presented in Appendix~\ref{sec: optimized Fuzzy PSI L p}.
Since our \Prot[FMap\text{-}prefix] achieves the same functionality as \Prot[FMap], we only need to redesign the remaining part of the fuzzy PSI protocol (i.e., the refined filtering phase). 

Similarly, we use the prefix representation of the interval of size $2\delta+1$ to reduce the number of key-value pairs for the refined filtering. 
With prefix optimization, the sender extends its every single input to $\mu$ prefixes in order to match the prefixes of the receiver. Therefore, both parties obtain $\mu$ outputs from the so-OPPRF for each dimension of a single element where at most one of these $\mu$ prefixes will match. We then use our EQSel protocol again to select the matched prefix, and the output is either the associated value or a random value.

The detailed protocol is shown in Figure~\ref{Prot:fpsi1 prefix}. 
The security proof of protocol \ensuremath{\Pi_\mathsf{FPSI\text{-}Prefix}^{L_\infty}}\xspace follows that of protocol \ensuremath{\Pi_\mathsf{FPSI}^{L_\infty}}\xspace in Section~\ref{sec: Fuzzy-PSI l infty}. We omit it due to limited space.

\section{Evaluation}
\label{sec: Evaluation}

\subsection{Experimental Setup}

We implement our protocols in C++ and our code is {available at \color{blue}\url{https://github.com/Th0masAndy/FPSI}}. 
All experiments are conducted on a server running Ubuntu~22.04, equipped with two AMD EPYC~9555 64-core processors and 512~GB of RAM. 
The sender and receiver are emulated as two separate threads within a single process.
Each experiment is repeated five times, and the average values are reported. 
We set the computational security parameter to $\kappa = 128$ and the statistical security parameter to $\lambda = 40$. 
Network conditions are simulated using the Linux \texttt{tc} command, configured with a bandwidth of 10~Gbps and a latency of 0.02~ms. 
In Appendix~\ref{sec: wan}, we further evaluate the performance under varying bandwidth and latency settings.

We compare our protocols with the state-of-the-art linear-complexity fuzzy PSI schemes~\cite{gao2025efficient, dang2025ccs}, as they rely on the same assumption as ours and support high-dimensional inputs.
For a fair comparison, we re-execute their publicly available implementations under the same experimental environment as ours.
Below, we detail the implementation components of our protocols.
For OKVS, we adopt the implementation from~\cite{raghuraman2022blazing}.
The implementations of si-OPRF and so-OPRF are built upon the MPC-friendly PRF proposed by~\cite{alamati2024improved}.
We use silent OT from the libOTe library~\cite{libOTe}, and the implementations of ssPEQT and PEQT are based on the protocols from~\cite{rindal2021vole}. For the prefix-related algorithms, we use the implementation of~\cite{dang2025ccs}.

\subsection{Performance of Fuzzy PSI}

We evaluate the performance of our fuzzy PSI protocols in Section \ref{sec: Fuzzy-PSI} and compare them with the state-of-the-art protocols~\cite{gao2025efficient,dang2025ccs}. The results are presented in Table~\ref{tab: low L0}, where we set the input sizes to $m = n = 2^8, 2^{12}, 2^{16}$, the dimension to $d = 4, 8, 16$, and the distance threshold to $\delta = 16, 32$.
We do not use prefix optimizations, since for these small distance thresholds our non-prefix protocols are more efficient.
The performance of prefix optimizations for large thresholds is presented in the next section.

As shown in Table~\ref{tab: low L0}, our protocol for $L_\infty$ distance achieves a $ {3{\sim}13\times}$ improvement in communication efficiency and a $ {9{\sim}145\times}$ improvement in computational efficiency, outperforming all previous protocols. Specifically, when $m = n = 2^{16}$, $d = 8$, and $\delta = 16$, our protocol consumes $ {13\times}$ less communication than~\cite{dang2025ccs} and $ {7\times}$ less than~\cite{gao2025efficient}. When $m = n = 2^{16}$, $d = 8$, and $\delta = 32$, our protocol is $ {50\times}$ faster than~\cite{dang2025ccs} and $ {145\times}$ faster than~\cite{gao2025efficient}.
We note that our protocols achieve better improvement for larger-scale inputs, since the overhead of the underlying silent OT ~\cite{boyle2019efficient} can be amortized.

We also present the results for $L_1$ and $L_2$ distances, which are commonly used metrics.
Compared with~\cite{dang2025ccs}, our protocol achieves up to an $ {80\times}$ improvement in running time and a $ {19\times}$ reduction in communication cost.
Moreover, our protocols for $L_1$ and $L_2$ distances incur only a minor overhead compared to that of $L_\infty$ distance, whereas~\cite{dang2025ccs} nearly doubles its overall cost.

Under limited network conditions, where the time for communication dominates the total running time, we provide the results in Appendix~\ref{sec: wan}.

\begin{table*}[t]
\centering
\caption{The communication (MB) and running time (s) of fuzzy PSI with prefix optimizations for large threshold $\delta$. We fix $m=n=2^{12}$ and $d=8$.
“---” indicates out-of-memory.
“Ours” denotes our fuzzy PSI protocols in Section~\ref {sec: Fuzzy-PSI} and “Ours-prefix” denotes our protocols with prefix optimizations in Section~\ref {sec: Fuzzy PSI with Prefix Optimization}. The best results are marked in {\color{green!90}green}.}
\label{tab:high delta}
\resizebox{0.70\textwidth}{!}{%
\begin{tabular}{cccccccccc}
\hline
\multirow{2}{*}{Metric}     & \multirow{2}{*}{Protocol}          & \multicolumn{2}{c}{$\delta=16$} & \multicolumn{2}{c}{$\delta=64$} & \multicolumn{2}{c}{$\delta=256$} & \multicolumn{2}{c}{$\delta=1024$} \\ \cline{3-10} 
                            &                                    & Comm.           & Time          & Comm.           & Time          & Comm.           & Time           & Comm.            & Time           \\ \hline
\multirow{4}{*}{$L_\infty$} & Gao et al.~\cite{gao2025efficient}        & 738.5           & 151.5         & 2876.9          & 496.6         & 11430.5         & 1850.6         & ---              & ---            \\
                            & Dang et al.~\cite{dang2025ccs} & 1389.5          & 86.2          & 1605.1          & 95.6          & 2581.0          & 156.9          & 2795.0           & 164.6          \\
                            & Ours                               & \cellcolor{green!25}112.3           & \cellcolor{green!25}1.8           & 308.3           & \cellcolor{green!25}2.6           & 1093.1          & \cellcolor{green!25}7.2            & 4235.8           & 19.7           \\ 
                            & Ours-prefix                        & 184.8           & 5.4           & \cellcolor{green!25}225.7           & 5.7           & \cellcolor{green!25}289.1           & 7.9            & \cellcolor{green!25}366.2            & \cellcolor{green!25}10.0           \\ \hline
\multirow{4}{*}{$L_1$}      & Gao et al.~\cite{gao2025efficient}        & 740.1           & 143.4         & 2879.3          & 483.3         & 11433.6         & 1840.1         & ---              & ---            \\
                            & Dang et al.~\cite{dang2025ccs} & 1954.8          & 113.4         & 2487.7          & 152.6         & 3695.2          & 216.6          & 4479.7           & 241.5          \\
                            & Ours                               & \cellcolor{green!25}122.5           & \cellcolor{green!25}2.1           & 318.5           & \cellcolor{green!25}2.9           & 1103.5          & \cellcolor{green!25}7.0            & 4246.5           & 20.3           \\ 
                            & Ours-prefix                        & 252.2           & 7.0           & \cellcolor{green!25}295.9           & 7.2           & \cellcolor{green!25}367.7           & 9.6            & \cellcolor{green!25}456.6            & \cellcolor{green!25}10.9           \\\hline
\multirow{4}{*}{$L_2$}      & Gao et al.~\cite{gao2025efficient}        & 741.6           & 144.3         & 2881.5          & 490.6         & 11436.6         & 1856.3         & ---              & ---            \\
                            & Dang et al.~\cite{dang2025ccs} & 2239.4          & 126.6         & 2818.4          & 175.5         & 4248.0          & 303.1          & 5572.5           & 572.7          \\
                            & Ours                               & \cellcolor{green!25}122.7           & \cellcolor{green!25}2.1           & \cellcolor{green!25}319.4           & \cellcolor{green!25}2.9           & 1104.6          & \cellcolor{green!25}7.3            & 4247.8           & 21.0           \\
                            & Ours-prefix                        & 385.2           & 9.7           & 440.3           & 9.7           & \cellcolor{green!25}534.0           & 11.9           & \cellcolor{green!25}682.6            & \cellcolor{green!25}15.1           \\
                            \hline
\end{tabular}%
}
\end{table*}

\subsection{Performance of Fuzzy PSI with Prefix Optimization}

We report the performance of our fuzzy PSI protocols with prefix optimization for large distance threshold $\delta$.
We set relatively large $\delta$ as $\{2^4, 2^6, 2^8, 2^{10}\}$.
As shown in Table~\ref{tab:high delta}, our optimized protocols exhibit desirable performance gains for large distance thresholds $\delta$. Overall, compared with the state-of-the-art work~\cite{dang2025ccs}, we achieve a $ {7{\sim}10\times}$ improvement in communication efficiency and a $ {13{\sim}38\times}$ improvement in computational efficiency.
We note that our non-prefix protocols still exhibit competitive performance compared to~\cite{dang2025ccs}, especially in computation efficiency.

We also compare our optimized protocols with our non-prefix protocols. For $\delta \ge 64$, our optimized protocols significantly reduce the communication overhead. For instance, when $\delta = 1024$, it lowers the communication cost by $ {11\times}$ compared to the non-prefix protocols, while reducing the running time by approximately half.
However, our non-prefix protocols still achieve better runtime performance for $\delta \le 256$, since the optimized protocols introduce additional operations such as conditional selection.

We provide the performance under different network settings in Appendix~\ref{sec: wan}. Since our optimized protocol significantly reduces communication overhead, its overall efficiency further improves under limited network conditions.

\section{Discussion}

\noindent\textbf{Limitations.}
Our work still has the following limitations. 
First, similar to prior fuzzy PSI protocols~\cite{van2024fuzzy,van2025,gao2025efficient,dang2025ccs,piske2025distance}, our protocol relies on certain assumptions about the input distribution. These assumptions are typically parameterized by the distance threshold $\delta$. While such parameters are well-defined in theoretical analysis, determining appropriate values for real-world datasets remains challenging. Second, it remains unclear how to obtain malicious security without substantially compromising efficiency. A natural approach is to instantiate our construction using maliciously secure OT~\cite{yang2020ferret,roy2022softspokenot} and authenticated secret sharing based on IT-MACs~\cite{bendlin2011semi,nielsen2012new}. However, a key technical challenge is how to efficiently verify that the OKVS encoding is well-formed and consistent.

\noindent\textbf{Future Work.}
We outline several directions for future research. 
First, it is of interest to design fuzzy PSI protocols under weaker or more flexible assumptions, such as relaxing distribution constraints or supporting one-sided assumptions, while maintaining comparable efficiency. 
Second, although existing fuzzy PSI protocols guarantee correctness under ideal assumptions, it remains an open problem to quantify and evaluate the impact when these assumptions are violated in practice (e.g., when elements are closer than expected or exhibit collisions). Developing robustness and accuracy guarantees in such settings is an important direction for future work.

\section{Conclusion}
\label{Conclusion}

In this work, we present a novel modular design for fuzzy PSI based on symmetric primitives. Our protocol 
is built upon a new building block termed so-OPPRF,
achieving linear complexity with respect to the input size $m$, $n$, the dimension $d$, and the distance threshold $\delta$. We further optimize the protocol by incorporating prefix techniques, and reduce its complexity to logarithmic in $\delta$. Experimental results demonstrate that our protocols outperform linear-complexity state-of-the-art works, significantly improving both computation and communication efficiency.

\section{Acknowledgment}

This work was supported by the Nanyang Technological University Centre in Computational Technologies for Finance (NTU-CCTF). It was also supported by Lee Kong Chian Chair Professorship, Singapore Management University. It was also supported by the National Research Foundation, Singapore, and Cyber Security Agency of Singapore under its National Cybersecurity R\&D Programme and CyberSG R\&D Cyber Research Programme Office. Any opinions, findings, and conclusions or recommendations expressed in this material are those of the author(s) and do not necessarily reflect the views of NTU-CCTF, National Research Foundation, Singapore, Cyber Security Agency of Singapore as well as CyberSG R\&D Programme Office, Singapore.

\section{Ethics considerations}

This work provides effective solutions for the fuzzy private set intersection task, encouraging researchers to pay more attention to the privacy of data analysis. The experiments in this paper are all based on public datasets and do not contain any personal or illegal information. We believe that our research was done ethically.

\section{LLM usage considerations}

This paper used the LLM solely to assist with grammar checking and language polishing. The LLM did not contribute any conceptual ideas, methodological innovations, experimental designs, or analytical insights to this work. All intellectual contributions originate from the authors. All data provided to the LLM for linguistic refinement contain no sensitive, personal, or ethically problematic information.

\bibliographystyle{IEEEtran}
\bibliography{sample-base}

\begin{appendices}

\section{Threat Model and Security Proof}

\subsection{Threat Model}
\label{sec: Threat Model}
Similar to prior works~\cite{van2024fuzzy,dang2025ccs,gao2025efficient,van2025}, we consider static semi-honest probabilistic polynomial-time (PPT) adversaries.
Namely, a PPT adversary \AAA passively corrupts either the sender \SSS or the receiver \RRR at the beginning of the protocol and honestly follows the protocol specification.
We use the standard simulation-based security definition for secure two-party computation.
Our construction invokes multiple sub-protocols, and we use the \textit{hybrid model} to describe them.
By convention, a protocol invoking a functionality \FFF is referred to as the \FFF-hybrid model.
We give the formal security definition as follows.

\begin{definition}
Let $\mathsf{view}_{\SSS}^{\Pi}(x, y)$ and $\mathsf{view}_{\RRR}^{\Pi}(x, y)$ be the views of $\SSS$ and $\RRR$ in a protocol $\Pi$, respectively, where $x$ is the input of $\SSS$ and  $y$ is the input of $\RRR$. Let $\mathsf{out}(x, y)$ be the protocol's output of both parties and $\mathcal{F}(x, y)$ be the functionality's output. $\Pi$ is said to securely compute a functionality $\mathcal{F}$ in the semi-honest model if for every PPT adversary \AAA there exists PPT simulators $\mathsf{Sim}_{\SSS}$ and $\mathsf{Sim}_{\RRR}$ such that for all inputs $x$ and $y$,
\begin{equation}
\begin{aligned}
    &\{\mathsf{view}_{\SSS}^{\Pi}(x, y), \mathsf{out}(x, y)\} \approx_c \{\mathsf{Sim}_{\SSS}(x, \mathcal{F}_{\SSS}(x, y)), \mathcal{F}(x, y)\}, \\
    &\{\mathsf{view}_{\RRR}^{\Pi}(x, y), \mathsf{out}(x, y)\} \approx_c \{\mathsf{Sim}_{\RRR}(x, \mathcal{F}_{\RRR}(x, y)), \mathcal{F}(x, y)\}. \nonumber
\end{aligned}
\end{equation}
\end{definition}

\subsection{Proof of so-OPPRF}
\label{subsec: Proof of so-OPPRF}

We give the detailed proof of Theorem \ref{thm: Proof of so-OPPRF}.

\begin{proof}
    \label{proof: Proof of so-OPPRF}

    If both parties behave honestly, correctness on programmed points $(q, z) \in L$ follows from the correctness of OKVS, since $y_i^\SSS + y_i^\RRR = \mathsf{OKVS.Decode}(D, q) + f_i^\SSS + f_i^\RRR = z - F_k(q) + F_k(q) = z$. On unprogrammed points $x \in X$, the OKVS encoding $D$ is independent from $F_k(x)$, and thus the value $y_i^\SSS + y_i^\RRR := F_k(x) + \mathsf{OKVS.Decode}(D, x)$ is indistinguishable from a uniformly random distribution, since $F_k(\cdot)$ is indistinguishable from a random function by definition of \Func[OPRF]. We next consider a corrupted sender or receiver.

    \textbf{Corrupted sender.} We show the simulator $\Sim[\SSS](L, \{y_i^\SSS\}_{i \in [n]}, \OOO^{F^\prime})$:
    \begin{enumerate}

        \item \Sim[\SSS] computes $D \leftarrow \mathsf{OKVS.Encode}(\{(q_j, z_j - F(q_j))\}_{j\in[m]})$, where $F(q_j)$ are randomly sampled. For each query $x \notin \{q_j\}_{j\in[m]}$, \Sim[\SSS] defines $F(x) := F^\prime(x) - \mathsf{OKVS.Decode}(D, x)$.

        \item \Sim[\SSS] appends $(\{f_i^\SSS := y_i^\SSS\}_{i \in [n]}, \OOO^{F}, \OOO^{F^\prime})$ in the view.

    \end{enumerate}

    We show that the view of adversary simulated by \Sim has the identical distribution as its view in the real-world execution. For programmed points $q_j$, it holds $F^\prime(q_j) = F(q_j) + \mathsf{OKVS.Decode}(D, q_j)$. For unprogrammed points $x$, the output $F(x) := F^\prime(x) - \mathsf{OKVS.Decode}(D, x)$ is uniformly random according to the independence property of OKVS.

    \textbf{Corrupted receiver.} We show the simulator $\Sim[\RRR](X, \{y_i^\RRR\}_{i \in [n]})$:
    \begin{enumerate}

        \item \Sim[\RRR] samples a random OKVS encoding $D$ on $m$ random key-value pairs and appends $D$ in the view.

        \item \Sim[\RRR] appends $\{f_i^\RRR\}_{i \in [n]}$ in the view, where $f_i^\RRR := y_i^\RRR - \mathsf{OKVS.Decode}(D, x_i)$ for $i \in [n]$. 
        
    \end{enumerate}

We show that the output by $\Sim[\RRR]$ is indistinguishable from the real protocol.
By the double obliviousness property of OKVS and the randomness of so-OPRF, the randomly sampled OKVS encoding $D$ is statistically indistinguishable from a real protocol execution. Moreover, $f^\RRR_i$ has the same distribution in both the real and simulated protocols such that $y^\RRR_i = f^\RRR_i + \mathsf{OKVS.Decode}(D, x_i)$.

\end{proof}

\begin{figure}[t]
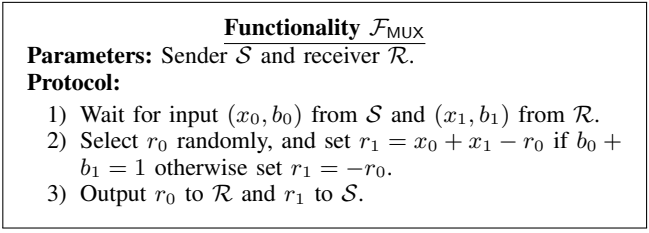

\begin{nffunc}{\Func[MUX]}

\noindent \textbf{Parameters:} Sender \SSS and receiver \RRR.

\noindent \textbf{Protocol:}

\begin{enumerate}

    \item Wait for input $(x_0,b_0)$ from \SSS and $(x_1,b_1)$ from \RRR.

    \item Select $r_0$ randomly, and set $r_1=x_0+x_1-r_0$ if $b_0+b_1=1$ otherwise set $r_1=-r_0$.
    \item Output $r_0$ to \RRR and $r_1$ to \SSS.
    
\end{enumerate}

\end{nffunc}
\vspace{0.5em}
\caption{Functionality of MUX.}
\label{Func:MUX}
\end{figure}

\begin{figure}[t]
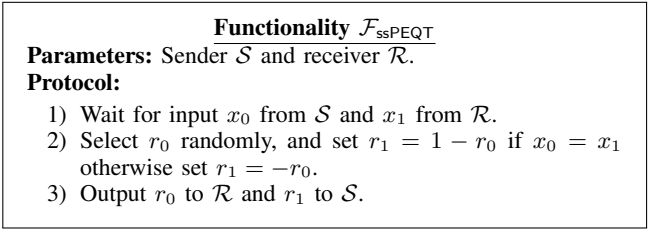

\begin{nffunc}{\Func[ssPEQT]}

\noindent \textbf{Parameters:} Sender \SSS and receiver \RRR.

\noindent \textbf{Protocol:}

\begin{enumerate}

    \item Wait for input $x_0$ from \SSS and $x_1$ from \RRR.

    \item Select $r_0$ randomly, and set $r_1=1-r_0$ if $x_0=x_1$ otherwise set $r_1=-r_0$.
    \item Output $r_0$ to \RRR and $r_1$ to \SSS.
    
\end{enumerate}

\end{nffunc}
\vspace{0.5em}
\caption{Functionality of secret-shared PEQT.}
\label{Func: ssPEQT}
\end{figure}

\begin{figure}[t]
    \begin{nfprot}{\Prot[getList_p]}
    \noindent \textbf{Parameters:} $\{\mathsf{ID}_{w_i}\}_{i\in[n]},\{w_i\}_{i\in[n]},p$. Denote $\mu = \log \delta$.
    
    \noindent \textbf{Protocol:}
    \begin{enumerate}
        \item For $i\in [n],k\in[d]$, compute $\{w_{0,i,k,h}\}_{h\in [\mu]} := \mathsf{Decompose}([w_{i,k}-\delta, w_{i,k}])$ and $\{w_{1,i,k,h}\}_{h\in [\mu]} := \mathsf{Decompose}([w_{i,k}+1, w_{i,k}+\delta])$.

        \item If $p = 1$: \begin{itemize}
                            \item $\mathsf{List_p} := \{(\mathsf{ID}_{w_i}\Vert k\Vert \sigma \Vert w_{\sigma,i,k,h}, 0 \Vert |w^*-w_{i,k}|)\}_{i\in [n],k\in[d],\sigma\in[0,1],h\in [\mu]}$, where $w^* := \mathsf{UpBound}(w_{\sigma,i,k,h})$ if $\sigma=0$ and $w^* := \mathsf{LowBound}(w_{\sigma,i,k,h})$ otherwise.
                        \end{itemize}
        
        \item If $p = 2$: \begin{itemize}
                            \item $\mathsf{List_p} := \{(\mathsf{ID}_{w_i}\Vert k\Vert \sigma \Vert w_{\sigma,i,k,h}, 0\Vert |w^*-w_{i,k}| \Vert |w^*-w_{i,k}|^2 )\}_{i\in [n],k\in[d],\sigma\in[0,1],h\in [\mu]}$, where $w^* := \mathsf{UpBound}(w_{\sigma,i,k,h})$ if $\sigma=0$ and $w^* := \mathsf{LowBound}(w_{\sigma,i,k,h})$ otherwise.
                        \end{itemize}

        \item Pad $\mathsf{List_p}$ into size $n\cdot d\cdot 2\cdot \mu$ and outputs $\mathsf{List_p}$
    \end{enumerate}
    
    \end{nfprot}
    \vspace{0.5em}
    \caption{Protocol for computing list for $p=1$ and $p=2$. The protocol can extend to arbitrary $p$.}
    \label{Prot:get list p}
\end{figure}

\begin{figure}[t]
    \begin{nfprot}{\Prot[getDistance_p]} 
    \noindent \textbf{Parameters:} Sender inputs $\{s^\SSS_h\}_{h\in [\mu]} \in \FF^{\mu \times p}$ , $ q\in \UU$, $\sigma\in \{0,1\}$. Receiver inputs $\{s^\RRR_h\}_{h\in [\mu]} \in \FF^{\mu \times p}$. Functionality $\Func[Mult]$. Denote $\mu = \log\delta$.
    
    \noindent \textbf{Protocol:}
    \begin{enumerate}
        \item \SSS computes $\{q_h\}_{h\in [\mu]} := \mathsf{AllPrefix}(q,\mu)$.
        \item For $h\in [\mu]$, \SSS computes $e_h := \abs{q-\mathsf{UpBound}(q_h)}$ if $\sigma=0$ and $e_h = \abs{q-\mathsf{LowBound}(q_h)}$ otherwise.
        \item If $p=1$: \begin{itemize}
                            \item For $h\in [\mu]$, \SSS computes $d_h^\SSS :=  e_{h} + s^\SSS_h$ and \RRR computes $d_h^\RRR := s^\RRR_h$. 
                        \end{itemize}
        \item If $p=2$: \begin{itemize}
                    \item For $h\in [\mu]$, \SSS parses $s^\SSS_h = s^\SSS_{h,1}\Vert s^\SSS_{h,2} \in \FF^2$ and \RRR parses $s^\RRR_h = s^\RRR_{h,1}\Vert s^\RRR_{h,2} \in \FF^2$.
                    \item For $h\in [\mu],i\in[1,p]$, \SSS and \RRR invoke $\Func[Mult]$ where \SSS inputs $e_h$ and \RRR inputs $s^\RRR_{h,1}$. \SSS receives $m^\SSS_{h}$ and \RRR receives $m^\RRR_{h}$.
                    \item For $h\in [\mu]$, \SSS computes $d_h^\SSS :=  e^2_{h} + 2e_h\cdot s^\SSS_{h,1} + 2m^\SSS_{h} + s^\SSS_{h,2}$ and $d_h^\RRR := 2m^\RRR_{h} + s^\RRR_{h,2}$.
                    \end{itemize}
        \item \SSS outputs $\{d_h^\SSS\}_{h\in [\mu] }$ and \RRR outputs $\{d_h^\RRR\}_{h\in [\mu]}$.
        
    \end{enumerate}
    
    \end{nfprot}
    \vspace{0.5em}
    \caption{Protocol for computing distance for $p=1$ and $p=2$. The protocol can extend to arbitrary $p$.}
    \label{Prot:sum distance}
\end{figure}

\subsection{Proof of Fuzzy PSI for $L_\infty$ distance}
\label{subsec: Proof of fuzzy PSI L-inf}

We give the detailed proof of Theorem \ref{thm: Proof of fuzzy PSI L-inf}.

\begin{proof}
\label{proof: Proof of fuzzy PSI L-inf}

    \textbf{Corrupted sender.} We show the simulator $\Sim[\SSS](Q)$, which invokes the sub-simulator $\Sim[\SSS]^{\mathsf{FMap}}$:
    \begin{enumerate}

        \item \Sim[\SSS] invokes $\Sim[\SSS]^{\mathsf{FMap}}(Q, \mathsf{ID}_Q)$ and appends the output to the view, where $\mathsf{ID}_Q$ is randomly sampled.

        \item \Sim[\SSS] randomly samples $\{ r^{\SSS}_{j,k} \}_{j\in [m],k\in [d]}$ and appends them to the view.
        
    \end{enumerate}

We show that the output by $\Sim[\SSS]$ is indistinguishable from the real protocol. 
According to the security property of fuzzy mapping in Theorem~\ref{thm: fuzzy mapping}, the simulated view of $\Sim[\SSS]^{\mathsf{FMap}}$ is indistinguishable from the real execution.
Besides, the only difference is $\{ r^{\SSS}_{j,k} \}_{j\in [m],k\in [d]}$. They are random secret shares in the real protocol according to the so-OPPRF functionality, while they are randomly sampled in the simulated view with the same distribution.
Therefore, the output by $\Sim[\SSS]$ is indistinguishable from the real protocol.

    \textbf{Corrupted receiver.}  We show the simulator $\Sim[\RRR](W, Z)$, which invokes the sub-simulator $\Sim[\RRR]^{\mathsf{FMap}}$:
    \begin{enumerate}
        \item \Sim[\RRR] randomly samples $\mathsf{ID}_W$. \Sim[\RRR] invokes $\Sim[\RRR]^{\mathsf{FMap}}(W, \mathsf{ID}_W)$ and appends the output to the view.

        \item \Sim[\RRR] randomly samples $\{ r^{\RRR}_{j,k} \}_{j\in [m],k\in [d]}$ and samples a random function $F$ such that $F(\mathsf{ID}_{w_i} \Vert k\Vert (w_{i,k}+t)) = 0$ for $t\in [-\delta,\delta], i \in [n],k\in [d]$.
        \Sim[\RRR] appends $(\{ r^{\RRR}_{j,k} \}_{j\in [m],k\in [d]}, \OOO^F)$ to the view.
        
        \item \Sim[\RRR] uses $\bot$ to pad $Z$ to $m$ elements, randomly shuffles the set $Z$, and sets $b_j := 0$ if $z_j = 0$ and $b_j := 1$ otherwise for $j \in [m]$. \Sim[\RRR] appends $(\{b_j\}_{j \in [m]}, Z)$ to the view.
    \end{enumerate}

    We show that the output by $\Sim[\RRR]$ is indistinguishable from the real protocol. According to the security property of fuzzy mapping in Theorem~\ref{thm: fuzzy mapping}, the simulated view of $\Sim[\RRR]^{\mathsf{FMap}}$ is indistinguishable from the real execution. Moreover, the way \RRR obtains the elements in $Z$ is identical to the real execution since the elements in $Z$ are randomly shuffled. Besides, the only difference is $(\{ r^{\RRR}_{j,k} \}_{j\in [m],k\in [d]}, \OOO^F)$. $\{r^{\RRR}_{j,k} \}_{j\in [m],k\in [d]}$ are random secret shares in the real protocol according to the so-OPPRF functionality, while they are randomly sampled in the simulated view with the same distribution.
    Then, the function $F$ is constructed in the same manner in the real and simulated executions. Therefore, the output by $\Sim[\RRR]$ is indistinguishable from the real protocol.
    
\end{proof}

\begin{figure}[t]
    \begin{nfprot}{\ensuremath{\Pi_\mathsf{FPSI\text{-}Prefix}^{L_p}}\xspace}
        \noindent \textbf{Parameters:} Distance threshold \( \delta \). Prefix parameters $\mu^\prime = \log\delta$ and $\mu = 1+\log\delta$.

        \noindent \textbf{Input:} \SSS inputs \( Q=\{q_j\}_{j\in [m]} \in \mathbb{U}^{d\times m} \) and \RRR inputs \( W=\{w_i\}_{i\in [n]} \in \mathbb{U}^{d\times n} \).

        \vspace{0.5em}
        \noindent \textbf{Protocol:}

        \begin{enumerate}
             \item \sender and \receiver invoke sub-protocol \Prot[FMap\text{-}Prefix], where \sender inputs $Q$ and \receiver inputs $W$. \sender obtains  $\{\mathsf{ID}_{q_j} \}_{j\in [m]}$ and \receiver obtains $\{\mathsf{ID}_{w_i} \}_{i\in [n]}$.

            \item \RRR invokes $\Prot[getList_p](\{\mathsf{ID}_{w_i}\}_{i\in[n]},\{w_i\}_{i\in[n]},p)$ and gets $\mathsf{List_p}$.
            \item For $j \in [m],k\in [d]$, \SSS computes $\{q_{j,k, h}\}_{h \in [\mu]} := \mathsf{AllPrefix}(q_{j,k},\mu)$.

            \item \sender and \receiver invoke functionality \Func[so\text{-}OPPRF], where \RRR inputs $\mathsf{List_p}$ and \SSS inputs $\{ \mathsf{ID}_{q_j} \Vert k \Vert \sigma \Vert q_{j,k,h} \}_{\sigma \in [0,1],j \in [m],k \in [d],h\in [\mu]}$. \SSS receives $\{ e^{\SSS}_{\sigma,j,k,h} \Vert v^{\SSS}_{\sigma,j,k,h} \}_{\sigma \in [0,1],j\in [m],k\in [d],h\in [\mu]}$ and \RRR receives $\{ e^{\RRR}_{\sigma,j,k,h} \Vert v^{\RRR}_{\sigma,j,k,h} \}_{\sigma \in [0,1],j\in[m],k\in [d],h\in [\mu]}$.

            \item For $\sigma \in [0,1],j\in [m],k\in [d],h\in [\mu]$, \SSS and \RRR invoke functionality {\Func[B2A], where \SSS inputs $v^{\SSS}_{\sigma,j,k,h}$ and \RRR inputs $v^{\RRR}_{\sigma,j,k,h}$. \SSS receives $s^{\SSS}_{\sigma,j,k,h}$ and \RRR receives $s^{\RRR}_{\sigma,j,k,h}$}.

            \item For $\sigma \in [0,1],j\in [m],k\in [d]$, \SSS and \RRR invoke $\Prot[getDistance_p]$, where \SSS inputs $(\{ s^{\SSS}_{\sigma,j,k,h}\}_{h\in [\mu]}, q_{j,k},\sigma)$, \RRR inputs $\{ s^{\RRR}_{\sigma,j,k,h}\}_{h\in [\mu]}$. 
            \SSS receives $\{d^\SSS_{\sigma,j,k,h}\}_{h\in[\mu]}$ and \RRR receives $\{d^\RRR_{\sigma,j,k,h}\}_{h\in[\mu]}$.

            \item For $j\in [m],k\in [d]$, \SSS and \RRR invoke $\Func[EQSel]$, where \SSS inputs $\{e^{\SSS}_{\sigma,j,k,h}, d^{\SSS}_{\sigma,j,k,h}\}_{\sigma \in [0,1],h\in [\mu]}$ and \RRR inputs $\{e^{\RRR}_{\sigma,j,k,h}, d^{\RRR}_{\sigma,j,k,h}\}_{\sigma \in [0,1],h\in [\mu]}$. \SSS receives $r^{\SSS}_{j,k}$ and \RRR receives $r^{\RRR}_{j,k}$.

            \item For $j\in [m]$, \SSS computes $r^{\SSS}_{j} := {\sum}_{k\in[d]}r^{\SSS}_{j,k}$, and \RRR computes $r^{\RRR}_{j} := {\sum}_{k\in [d]}r^{\RRR}_{j,k}$.

            \item For \( j \in [m] \), \SSS and \RRR invoke functionality { \Func[Interval], where \SSS inputs $r^{\SSS}_{j}$ and \RRR inputs $r^{\RRR}_{j}$. \RRR receives $b_j := \textbf{1}(r^{\SSS}_{j}+r^{\RRR}_{j}\le \delta^p)$}.
            \item For \( j \in [m] \), \RRR and \SSS invoke functionality \Func[OT], where \RRR inputs $b_j$ and \SSS inputs $(\bot ,q_j)$. \RRR receives OT outputs $\mathsf{z}_j$.
            \item \RRR outputs $Z := \{z_j \mid b_j = 1\;\text{for}\; j\in [m]\}$.
        \end{enumerate}
    \end{nfprot}
    \vspace{0.5em}
    \caption{Protocol of fuzzy PSI with prefix optimization for $L_p$ distance. }
    \label{Prot:fpsi2 prefix}
\end{figure}

\begin{table*}[t]
\centering
\caption{The communication (MB) and running time (s) of fuzzy PSI in Section~\ref{sec: Fuzzy-PSI} under different network settings.}
\label{tab: non-prefix wan}
\resizebox{\textwidth}{!}{%
\begin{tabular}{ccccccccccccccc}
\hline
\multirow{2}{*}{$m=n$}   & \multirow{2}{*}{$\delta$} & \multirow{2}{*}{$d$} & \multicolumn{4}{c}{$L_\infty$}    & \multicolumn{4}{c}{$L_1$}         & \multicolumn{4}{c}{$L_2$}         \\ \cline{4-15} 
                       &                           &                      & Comm.  & 10Gbps & 1Gbps & 100Mbps & Comm.  & 10Gbps & 1Gbps & 100Mbps & Comm.  & 10Gbps & 1Gbps & 100Mbps \\ \hline
\multirow{6}{*}{256}   & \multirow{3}{*}{16}       & 4                    & 7.3    & 0.3    & 0.4   & 4.8     & 7.8    & 0.3    & 0.5   & 5.2     & 7.8    & 0.4    & 0.6   & 5.2     \\
                       &                           & 8                    & 10.6   & 0.3    & 0.5   & 5.3     & 11.3   & 0.4    & 0.6   & 5.8     & 11.4   & 0.4    & 0.6   & 5.8     \\
                       &                           & 16                   & 17.1   & 0.4    & 0.7   & 5.8     & 18.4   & 0.5    & 0.8   & 6.2     & 18.4   & 0.5    & 0.7   & 6.2     \\ \cline{2-15} 
                       & \multirow{3}{*}{32}       & 4                    & 9.4    & 0.3    & 0.6   & 5.1     & 9.9    & 0.3    & 0.6   & 5.5     & 9.9    & 0.4    & 0.6   & 5.5     \\
                       &                           & 8                    & 14.7   & 0.4    & 0.7   & 5.7     & 15.4   & 0.5    & 0.7   & 6.4     & 15.5   & 0.4    & 0.7   & 6.4     \\
                       &                           & 16                   & 25.2   & 0.5    & 0.8   & 6.7     & 26.6   & 0.5    & 0.9   & 7.1     & 26.6   & 0.6    & 0.9   & 7.1     \\ \hline
\multirow{6}{*}{4096}  & \multirow{3}{*}{16}       & 4                    & 60.5   & 1.3    & 1.8   & 10.4    & 66.1   & 1.4    & 2.0   & 11.3    & 66.3   & 1.4    & 2.0   & 11.3    \\
                       &                           & 8                    & 112.3      & 2.0    & 2.7   & 15.1    & 122.5  & 2.1    & 3.0   & 16.5    & 122.7  & 2.2    & 3.0   & 16.5    \\ 
                       &                           & 16                   & 215.8     & 3.0    & 4.4   & 25.0    & 235.3  & 3.5    & 5.0   & 27.5    & 235.5  & 3.5    & 4.9   & 27.6    \\ \cline{2-15} 
                       & \multirow{3}{*}{32}       & 4                    & 93.2   & 1.4    & 2.2   & 13.4    & 98.7   & 1.6    & 2.4   & 14.2    & 99.2   & 1.6    & 2.4   & 14.3    \\
                       &                           & 8                    & 177.6      & 2.2    & 3.5   & 21.4    & 187.8  & 2.5    & 3.8   & 22.7    & 188.2  & 2.5    & 3.9   & 22.7    \\
                       &                           & 16                   & 346.5     & 3.7    & 6.1   & 37.0    & 366.0  & 4.1    & 6.7   & 40.3    & 366.4  & 4.0    & 6.6   & 39.2    \\ \hline
\multirow{6}{*}{65536} & \multirow{3}{*}{16}       & 4                    & 908.7  & 17.2   & 22.4  & 96.9    & 990.9  & 17.9   & 24.4  & 105.7   & 998.8  & 18.6   & 24.5  & 106.4   \\
                       &                           & 8                    & 1737.6 & 25.0   & 36.8  & 172.1   & 1893.8 & 27.8   & 39.2  & 190.4   & 1901.8 & 27.8   & 39.5  & 188.3   \\
                       &                           & 16                   & 3396.0 & 42.0   & 65.0  & 319.4   & 3700.2 & 46.9   & 71.0  & 351.1   & 3708.2 & 47.7   & 71.3  & 355.5   \\ \cline{2-15} 
                       & \multirow{3}{*}{32}       & 4                    & 1432.0 & 19.0   & 28.5  & 143.0   & 1516.2 & 20.2   & 30.5  & 151.1   & 1526.1 & 20.4   & 31.1  & 155.6   \\
                       &                           & 8                    & 2784.7 & 28.9   & 49.8  & 271.1   & 2942.9 & 31.5   & 53.9  & 284.3   & 2952.8 & 31.7   & 53.8  & 286.6   \\
                       &                           & 16                   & 5491.5 & 49.7   & 92.0  & 511.3   & 5797.7 & 55.8   & 99.2  & 541.9   & 5807.7 & 55.1   & 99.8  & 541.4   \\ \hline
\end{tabular}%
}
\end{table*}

\begin{table*}[t]
\centering
\caption{The communication (MB) and running time (s) of fuzzy PSI under different network settings. We fix $m=n=2^{12}$ and $d=8$. “Ours-prefix” denotes our optimized protocol in Section~\ref {sec: Fuzzy PSI with Prefix Optimization} and “Ours” denotes our non-optimized protocol in Section~\ref {sec: Fuzzy-PSI}.}
\label{tab: prefix wan}
\resizebox{\textwidth}{!}{%
\begin{tabular}{cccccccccccccc}
\hline
\multirow{2}{*}{$\delta$} & \multirow{2}{*}{Protocol} & \multicolumn{4}{c}{$L_\infty$}    & \multicolumn{4}{c}{$L_1$}         & \multicolumn{4}{c}{$L_2$}         \\ \cline{3-14} 
                          &                           & Comm.  & 10Gbps & 1Gbps & 100Mbps & Comm.  & 10Gbps & 1Gbps & 100Mbps & Comm.  & 10Gbps & 1Gbps & 100Mbps \\ \hline
\multirow{2}{*}{16}       & Ours                      & 112.3  & 1.8    & 2.6   & 15.1    & 122.5  & 2.1    & 3.0   & 16.4    & 122.7  & 2.1    & 3.0   & 16.4    \\
                          & Ours-prefix               & 184.8  & 5.4    & 6.3   & 27.4    & 252.2  & 7.0    & 8.4   & 35.6    & 385.2  & 9.7    & 11.6  & 49.3    \\ \hline
\multirow{2}{*}{64}       & Ours                      & 308.3  & 2.6    & 5.1   & 33.7    & 318.5  & 2.9    & 5.4   & 35.5    & 319.4  & 2.9    & 5.6   & 35.3    \\
                          & Ours-prefix               & 225.7  & 5.7    & 6.8   & 31.2    & 295.9  & 7.2    & 8.9   & 39.2    & 440.3  & 9.7    & 12.4  & 54.0    \\ \hline
\multirow{2}{*}{256}      & Ours                      & 1093.1 & 7.2    & 15.1  & 104.5   & 1103.5 & 7.0    & 15.7  & 109.3   & 1104.6 & 7.3    & 15.5  & 106.3   \\
                          & Ours-prefix               & 289.1  & 7.9    & 9.0   & 38.4    & 367.7  & 9.6    & 11.1  & 46.9    & 534.0  & 11.9   & 15.1  & 63.1    \\ \hline
\multirow{2}{*}{1024}     & Ours                      & 4235.8 & 19.7   & 55.3  & 389.7   & 4246.5 & 20.3   & 55.9  & 386.5   & 4247.8 & 21.0   & 56.0  & 388.9   \\
                          & Ours-prefix               & 366.2  & 10.0   & 9.7   & 42.8    & 456.6  & 10.9   & 13.1  & 55.5    & 682.6  & 15.1   & 18.9  & 78.7  \\ \hline
\end{tabular}%
}
\end{table*}

\subsection{Proof of Fuzzy PSI for $L_p$ distance}
\label{subsec: Proof of fuzzy PSI L-p}

We give the detailed proof of Theorem \ref{thm: Proof of fuzzy PSI L-p}, which is similar to that of Theorem \ref{thm: Proof of fuzzy PSI L-inf}.

\begin{proof}
\label{proof: Proof of fuzzy PSI L-inf}

    \textbf{Corrupted sender.} We show the simulator $\Sim[\SSS](Q)$, which invokes the sub-simulator $\Sim[\SSS]^{\mathsf{FMap}}$:
    \begin{enumerate}
        \item \Sim[\SSS] randomly samples $\mathsf{ID}_Q$. \Sim[\SSS] invokes $\Sim[\SSS]^{\mathsf{FMap}}(Q, \mathsf{ID}_Q)$ and appends the output to the view.

        \item \Sim[\SSS] randomly samples $\{ r^{\SSS}_{j,k} \}_{j\in [m],k\in [d]}$ and appends it to the view.

        \item \Sim[\SSS] randomly samples $\{ d^{\SSS}_{j,k} \}_{j\in [m],k\in [d]}$ and appends it to the view.
    \end{enumerate}

We show that the output by $\Sim[\SSS]$ is indistinguishable from the real protocol. According to the security property of fuzzy mapping in Theorem~\ref{thm: fuzzy mapping}, the simulated view of $\Sim[\SSS]^{\mathsf{FMap}}$ is indistinguishable from the real execution.
Besides, the only differences are $\{ r^{\SSS}_{j,k} \}_{j\in [m],k\in [d]}$ and $\{ d^{\SSS}_{j,k} \}_{j\in [m],k\in [d]}$. They are random secret shares in the real protocol according to the functionalities \Func[so\text{-}OPPRF] and \Func[B2A], while they are randomly sampled in the simulated view with the same distribution.
Therefore, the output by $\Sim[\SSS]$ is indistinguishable from the real protocol.

    \textbf{Corrupted receiver.}   We show the simulator $\Sim[\RRR](W, Z)$, which invokes the sub-simulator $\Sim[\RRR]^{\mathsf{FMap}}$:
    \begin{enumerate}
        \item \Sim[\RRR] randomly samples $\mathsf{ID}_W$. \Sim[\RRR] invokes $\Sim[\RRR]^{\mathsf{FMap}}(W, \mathsf{ID}_W)$ and appends the output to the view.

        \item \Sim[\RRR] randomly samples $\{ r^{\RRR}_{j,k} \}_{j\in [m],k\in [d]}$ and samples a random function $F$ such that $F(\mathsf{ID}_{w_i} \Vert k\Vert (w_{i,k}+t)) = \abs{t}^p$ for $t\in [-\delta,\delta], i \in [n],k\in [d]$.
        \Sim[\RRR] appends $(\{ r^{\RRR}_{j,k} \}_{j\in [m],k\in [d]}, \OOO^F)$ to the view.

        \item \Sim[\RRR] randomly samples $\{ d^{\RRR}_{j,k} \}_{j\in [m],k\in [d]}$ and appends them to the view.
        
        \item \Sim[\RRR] uses $\bot$ to pad $Z$ to $m$ elements, randomly shuffles the set $Z$, and sets $b_j := 0$ if $z_j = 0$ and $b_j := 1$ otherwise for $j \in [m]$. \Sim[\RRR] appends $(\{b_j\}_{j \in [m]}, Z)$ to the view.
        
    \end{enumerate}

    We show that the output by $\Sim[\RRR]$ is indistinguishable from the real protocol. According to the security property of fuzzy mapping in Theorem~\ref{thm: fuzzy mapping}, the simulated view of $\Sim[\RRR]^{\mathsf{FMap}}$ is indistinguishable from the real execution. Moreover, the way \RRR obtains the elements in $Z$ is identical to the real execution since the elements in $Z$ are randomly shuffled. Besides, the only difference is $(\{ r^{\RRR}_{j,k} \}_{j\in [m],k\in [d]}, \OOO^F, \{ d^{\RRR}_{j,k} \}_{j\in [m],k\in [d]})$. $\{r^{\RRR}_{j,k} \}_{j\in [m],k\in [d]}$ and $\{ d^{\RRR}_{j,k} \}_{j\in [m],k\in [d]}$ are random secret shares in the real protocol according to the functionalities \Func[so\text{-}OPPRF] and \Func[B2A], while they are randomly sampled in the simulated view with the same distribution.
    Then, the function $F$ is constructed in the same manner in the real and simulated executions. Therefore, the output by $\Sim[\RRR]$ is indistinguishable from the real protocol.
    
\end{proof}

\section{Details of Fuzzy PSI with Prefix Optimizations}

\subsection{Other Building Blocks}
\label{sec: mpc func}
We present the functionalities of \Func[MUX] and \Func[ssPEQT] in Figure~\ref{Func:MUX} and Figure~\ref{Func: ssPEQT}, respectively.

\subsection{Sub-protocols for \ensuremath{\Pi_\mathsf{FPSI\text{-}Prefix}^{L_p}}\xspace}

We present the sub-protocols of our optimized fuzzy PSI protocol \ensuremath{\Pi_\mathsf{FPSI\text{-}Prefix}^{L_p}}\xspace in Figure~\ref{Prot:get list p} and Figure~\ref{Prot:sum distance}.

\subsection{Optimized Fuzzy PSI for $L_p$ Distance}
\label{sec: optimized Fuzzy PSI L p}

With prefix optimization, it is non-trivial to convert $L_\infty$ distance to $L_p$ distance.
As the interval is broken into prefixes, we cannot directly compute the distance for every point within the interval; instead, we only know the distance associated with each prefix. Dang et al.~\cite{dang2025ccs} addressed this issue by decomposing the distance into two parts. Specifically, the distance between points $x$ and $y$ can be expressed as $\abs{x - x^*} + \abs{x^* - y}$, where $x^*$ represents either the upper or lower bound of the prefix that $x$ falls into. The value $\abs{x - x^*}$ can be computed by the sender, while $\abs{x^* - y}$ can be computed by the receiver, and the two parties then jointly compute the sum of these two distances. 

We adopt a similar method but realize the same functionality using more efficient MPC primitives to keep consistent with our framework, whereas Dang et al.~\cite{dang2025ccs} heavily rely on Paillier encryption. Our modified construction for the key-value pairs is presented in Figure~\ref{Prot:get list p} and Figure~\ref{Prot:sum distance}.

The detailed protocol for \ensuremath{\Pi_\mathsf{FPSI\text{-}Prefix}^{L_p}}\xspace is shown in Figure~\ref{Prot:fpsi2 prefix} and its security proof follows Appendix~\ref{subsec: Proof of fuzzy PSI L-p}.

\section{Performance under Different Network Settings}
\label{sec: wan}

We provide the performance of our fuzzy PSI protocols and the variants with prefix optimizations under different network settings in Table~\ref{tab: non-prefix wan} and Table~\ref{tab: prefix wan}. We set the network to 1 Gbps bandwidth with 40 ms latency and 100 Mbps with 80 ms latency.

\clearpage %

\section{Meta-Review}

The following meta-review was prepared by the program committee for the 2026
IEEE Symposium on Security and Privacy (S\&P) as part of the review process as
detailed in the call for papers.

\subsection{Summary}
The paper improves the efficiency of a known construction of Fuzzy-PSI where two parties match similar but not necessarily identical elements. It contributes the primitive of an Oblivious Programmable Pseudo-Random Function (OPPRF) with secret shared outputs.

\subsection{Scientific Contributions}
\begin{itemize}
\item Provides a Valuable Step Forward in an Established Field.
\end{itemize}

\subsection{Reasons for Acceptance}
\begin{enumerate}
\item The paper provides a valuable step forward in the established field of Fuzzy-PSI. It significantly improves the performance of a known construction by contributing a new primitive - an OPPRF with secret shared outputs - which may be of independent interest.
\end{enumerate}

\end{appendices}

\end{document}